\documentclass[11pt,letterpaper]{article}
\usepackage{amsmath,amsfonts,amsthm,amssymb}
\usepackage{setspace}
\usepackage{fancyhdr}
\usepackage{lastpage}
\usepackage{extramarks}
\usepackage{xspace}
\usepackage{chngpage}

\usepackage{algorithm,algorithmicx}
\usepackage[noend]{algpseudocode}
\usepackage{soul,color}
\usepackage{graphicx,float,wrapfig}
\usepackage[font=small,labelfont=bf]{caption}
\usepackage[margin=1in]{geometry}
\usepackage{enumitem}
\usepackage{thmtools,thm-restate}
\usepackage{bbm}

\setlist{nosep}

\allowdisplaybreaks

\usepackage{mdframed}
\newtheorem*{mdresult}{Result}

\usepackage[usenames,dvipsnames]{xcolor}
\usepackage[linktocpage=true,breaklinks,colorlinks,citecolor=blue,linkcolor=BrickRed]{hyperref}

\usepackage{wrapfig}
  
\title{\textmd{\bf Learning Optimal Posted Prices for a Unit-Demand Buyer
  }}

\date{\today}
\author{Yifeng Teng\footnote{Google Research. Email: yifengt@google.com.}\and  Yifan Wang\footnote{School of Computer Science, Georgia Institute of Technology, Atlanta, GA, USA. Email: ywang3782@gatech.edu. Supported in part by NSF awards CCF-2327010 and CCF-2440113.}}

\usepackage[usenames,dvipsnames]{xcolor}
\usepackage[linktocpage=true,breaklinks,colorlinks,citecolor=blue,linkcolor=BrickRed]{hyperref}
\usepackage{cleveref}

\DeclareUnicodeCharacter{2217}{*}

\newtheorem{Theorem}{Theorem}[section]
\newtheorem{Lemma}[Theorem]{Lemma}

\newtheorem{Definition}[Theorem]{Definition}

\newtheorem{Claim}[Theorem]{Claim}

\newcommand{\parta}{(\text{\uppercase\expandafter{\romannumeral1}})}
\newcommand{\partb}{(\text{\uppercase\expandafter{\romannumeral2}})}
\newcommand{\partc}{(\text{\uppercase\expandafter{\romannumeral3}})}

\newcommand{\rev}{\mathsf{Rev}}

\newcommand{\E}{\mathbb{E}}

\newcommand{\OTild}{\widetilde{O}}
\newcommand{\one}{\mathbf{1}}%
\renewcommand{\Pr}[2][]{\mbox{\rm\bf Pr}_{#1}\left[#2\right]}%

\newcommand{\pr}{\mathbf{Pr}} 

\newcommand{\ignore}[1]{{}}

\newcommand{\bD}{\boldsymbol{\mathrm{D}}}
\newcommand{\tbD}{\widetilde{\boldsymbol{\mathrm{D}}}}
\newcommand{\tbG}{\widetilde{\boldsymbol{\mathrm{G}}}}
\newcommand{\tbH}{\widetilde{\boldsymbol{\mathrm{H}}}}
\newcommand{\tD}{\widetilde{D}}
\newcommand{\tH}{\widetilde{H}}
\newcommand{\bG}{\boldsymbol{\mathrm{G}}}
\newcommand{\bH}{\boldsymbol{\mathrm{H}}}
\newcommand{\bbD}{\breve{\boldsymbol{\mathrm{D}}}}
\newcommand{\ttheta}{\tilde{\theta}}
\newcommand{\bE}{\boldsymbol{\mathrm{E}}}
\newcommand{\SUDPP}{BUPP~}
\newcommand{\bp}{\boldsymbol{p}}

\newcommand{\poly}{\mathsf{poly}}

\newcommand{\bv}{\boldsymbol{v}}

\newcommand{\eps}{\epsilon}

\newcommand{\erev}{\mathsf{ExAnteRev}}

\newcommand{\IGNORE}[1]{}

\begin{document}

\setlength{\abovedisplayskip}{2pt}
\setlength{\belowdisplayskip}{2pt}

\maketitle \thispagestyle{empty}

\begin{abstract}

We study the problem of learning the optimal item pricing for a unit-demand buyer with independent item values, and the learner has query access to the buyer's value distributions. We consider two common query models in the literature: the sample access model where the learner can obtain a sample of each item value, and the pricing query model where the learner can set a price for an item and obtain a binary signal on whether the sampled value of the item is greater than our proposed price. In this work, we give nearly tight sample complexity and pricing query complexity of the unit-demand pricing problem.

\end{abstract}

\newpage

\section{Introduction}

Multi-item mechanism design has been studied extensively in the literature of economics and computation in the past two decades. Revenue optimal mechanisms are known to be extremely complicated even in the single-buyer setting, with many unnatural properties such as randomness \cite{thanassoulis2004haggling}, requiring an infinite number of menus to describe \cite{briest2015pricing,hart2019selling,daskalakis2013mechanism}, hardness to compute \cite{daskalakis2014complexity,chen2015complexity}, revenue non-monotonicity \cite{hart2015maximal}, and revenue discontinuity \cite{psomas2019smoothed, chawla2020menu}. As a result of such complexity in the analysis and implementation of the optimal mechanisms, many recent works in the multi-dimensional mechanism design setting have been focused on approximating the optimal revenue via simple mechanisms. One particularly important simple mechanism that is both studied in the literature and implemented in real-world scenarios is the item pricing mechanism. 

In this paper, we study the item pricing mechanisms for a single unit-demand buyer who wants to obtain at most one item. A monopoly seller has $n$ items to sell to a single buyer with quasi-linear utility. In an item pricing mechanism, the seller posts a price $p_i$ for each item $i\in[n]$ to the buyer. The buyer has bounded value $v_i\in[0,1]$ for each item $i$, and given the posted prices, the buyer chooses the item $i\in[n]$ with the highest non-negative utility $v_i-p_i$. We assume that the buyer has the item values $\bv=(v_1,\cdots,v_n)$ drawn from a product distribution $\bD=D_1\times D_2\times\cdots\times D_n$. Compared to the optimal revenue-maximizing mechanism, which may require displaying a menu of options of even infinite size, the item pricing mechanism is considerably simpler to implement in practice. Furthermore, when selling to a unit-demand buyer, it's known that item pricing achieves at least $1/3$ of the optimal revenue \cite{JL-FOCS24}. This simplicity and guaranteed performance make it a compelling choice.

Traditional research on Bayesian revenue maximization usually assumes that the seller has full information on the buyer's underlying value distribution. Recent literature has been studying more realistic models about how the seller gets access to the buyer's value information. One particularly important model introduced by \cite{cole2014sample} is the \textit{sample access} model: instead of directly observing buyer's value distribution $\bD$, the seller can get i.i.d. samples $\bv_1,\bv_2,\cdots,\bv_N\sim \bD$ from buyer's underlying value distribution. The \textit{sample complexity} of the revenue maximization problem is defined by the number of samples $m$ needed for the seller to learn a mechanism, such that under such a mechanism, with high probability the expected payment collected from the buyer is at most $\eps$ away from the optimal revenue with accurate prior information. The sample complexity has been studied in both single-parameter (i.e. single-item or multi-unit auctions) settings and multi-parameter settings. While the sample complexity of single-parameter auctions has been well-studied (Eg. \cite{guo2019settling}), only polynomial upper bounds on the sample complexity of multi-parameter revenue maximization settings have been known.

Another important model that has attracted much attention recently is the \textit{pricing query access} model \cite{leme2023pricing}. The model is defined for a single-dimensional setting where the buyer has a value distribution of $D$ for a single item. In each round of query interaction, an independent sample $v\sim D$ is drawn from the value distribution, with the value unknown to the learner. The learner proposes a price $p$ and observes a binary outcome $\one[v\geq p]$ on whether the buyer with value $v$ can afford to purchase the item at price $p$. The \textit{pricing query complexity} of the revenue maximization problem is defined by the number of queries needed for the learner to learn a mechanism with revenue at most $\eps$ away from the optimal revenue of a mechanism with the exact value distribution. Such a query model has been studied for learning the statistics of the underlying distribution under different names, such as \textit{pricing queries} \cite{leme2023pricing}, \textit{threshold queries} \cite{okoroafor2023non}, and \textit{comparison feedback} \cite{meister2021learning}, but the natural extension to multi-parameter settings have never been done before. In this paper, for consistency with different literature and for simplicity of presentation, we use ``\textit{query}'' and ``\textit{query complexity}'' interchangeably with ``\textit{pricing query}'' and ``\textit{pricing query complexity}'' when the context is clear.

\subsection{Our Results}

In this work, we study the sample complexity and the pricing query complexity of learning the optimal posted prices for a unit-demand buyer with independent item values. As the first main result, we give an asymptotically tight sample complexity result (up to a polylogarithmic factor in $n$ and $\eps$) for learning the optimal item pricing mechanism for a unit-demand buyer. The result improves over a previous $\tilde{O}\left(\frac{n^2}{\eps^2}\right)$ bound via pseudo-dimension \cite{morgenstern2016learning,cai2017learning}, and an $\tilde{O}\left(\frac{n}{\eps^4}\right)$ bound that also applies to optimal randomized mechanisms \cite{GHTZ-COLT21}.
As far as we know this is the first tight sample complexity result for a revenue maximization problem in the multi-parameter setting.

\begin{Theorem}
\label{thm:main-sample}
For any distributions $\bD=D_1\times D_2\times\cdots\times D_n$ on $[0,1]^{n}$, $\tilde{\Theta}\left(\frac{n}{\eps^2}\right)$ samples over $\bD$ are necessary and sufficient to learn a posted pricing mechanism $\bp$ with revenue at most $\eps$ less than the optimal item pricing. 
\end{Theorem}

Our second main result is about the pricing query complexity of the unit-demand pricing problem. We first need to define the query model in the multi-parameter setting. In our proposed query model, in each query, the learner can propose a pair $(i,p)$ including the item ID and the proposed price. A value $v_i\sim D_i$ is drawn from item $i$ and a binary signal $\one[v_i\geq p]$ is released to the learner. We provide the first study of the pricing query complexity in the context of multi-parameter revenue maximization problems and obtain a strong query complexity upper bound for the unit-demand pricing problem. Notice that we are assuming the learner can only make a threshold query to a single-dimensional distribution in each round.

\begin{Theorem}
\label{thm:main-query}
For any distributions $\bD=D_1\times D_2\times\cdots\times D_n$ on $[0,1]^{n}$, $\tilde{O}\left(\frac{n^2}{\eps^3}\right)$ pricing queries over $\bD$ are sufficient to learn a posted pricing mechanism $\bp$ with revenue at most $\eps$ less than the optimal item pricing. 
\end{Theorem}

We conjecture that the query complexity in Theorem~\ref{thm:main-query} is tight, by observing the following strong evidence that $\Omega\left(\frac{n^2}{\eps^3}\right)$ is the lower bound of the query complexity of the ex-ante relaxation of the problem. In an ex-ante unit-demand pricing problem, the seller wants to compute an item pricing $\bp=(p_1,\cdots,p_n)$, such that the total probability of an item being sold $\sum_i\Pr{v_i\geq p_i}\leq 1$ is at most 1 in expectation, and maximizes the ex-ante revenue $\sum_i p_i\Pr{v_i\geq p_i}$.

\begin{Theorem}
\label{thm:exante-query}
For any distributions $\bD=D_1\times D_2\times\cdots\times D_n$ on $[0,1]^{n}$, $\Omega\left(\frac{n^2}{\eps^3}\right)$ pricing queries over $\bD$ are necessary to learn a posted pricing mechanism $\bp$ with revenue at most $\eps$ less than the optimal item pricing in the ex-ante relaxation. 
\end{Theorem}

Intuitively the ex-ante version is a simpler problem compared to the original unit-demand pricing problem as it does not take the preference of the agent into account. The upper bound algorithm of the original problem also works for the ex-ante version using a similar analysis, so $\tilde{\Theta}\left(\frac{n^2}{\eps^{3}}\right)$ is actually the \textit{tight} query complexity of the ex-ante unit-demand pricing problem.

\subsection{Technical Overview}

\paragraph{Sample Complexity Upper Bound} For the sample complexity upper bound, we apply the idea of \textit{strong revenue monotonicity} from \cite{GHTZ-COLT21} and \textit{approximate strong revenue monotonicity} from \cite{CHMY-EC23}. \cite{GHTZ-COLT21} develop a theory on the sample complexity of a learning problem on product distributions and prove that if the optimal objective strictly increases when the underlying distribution is replaced by a distribution that stochastically dominates it, the sample complexity of the optimization problem is $\tilde{\Theta}(n\eps^{-2})$. However, the unit-demand pricing problem is not strongly revenue monotone.  Consider a concrete example involving two items, $A$ and $B$. The buyer's valuation for item $A$ is $0.5$. For item $B$, the valuation is $1$ with probability $0.5$ and $0$ with probability $0.5$. The optimal pricing in this scenario is to set $p_A = 0.5$ and $p_B = 1$, resulting in an expected revenue of $0.75$. Next, let's substitute item $A$ with item $\tilde{A}$, for which the buyer values it at $0.6$ with probability $0.1$ and $0.5$ with probability $0.9$. While item $\tilde{A}$'s valuation distribution clearly dominates item $A$'s, the expected revenue decreases to $0.1 \cdot 0.5 + 0.9 \cdot (0.5 \cdot 1 + 0.5 \cdot 0.5) = 0.725$. This demonstrates that the problem is not strongly monotone. \footnote{Indeed, setting the price for item $A$/$\tilde{A}$ at $0.5$ and for item $B$ at $1$ gives the revenue-optimal mechanism for selling either items $A$ and $B$ or items $\tilde A$ and $B$. Consequently, even weak monotonicity fails to hold when selling multiple items to a unit-demand buyer with independent valuations, regardless of whether the optimal item pricing or the optimal truthful mechanism is applied.}

Although the strong revenue monotonicity does not hold for our problem, we are able to prove a weaker approximate strong revenue monotonicity result that for any two value distributions $\bD$ and $\bD'$ such that $D_i\succeq D_i'$ and $D_i\approx D'_i$ for every $i\in[n]$,\footnote{Here the ordering is first-order stochastic dominance, and the closeness is measured by Kolmogorov distance.} the revenue of any item pricing on $\bD$ is not much worse than $\bD'$. With the weaker revenue monotonicity result, $\tilde{O}(n\eps^{-2})$ queries are still sufficient. Our key technical contribution is to prove the approximate strong revenue monotonicity for the unit-demand pricing problem. The main idea is to carefully bound the change of the probability of each item being the utility-maximizing item, and show that when the item values slightly increase, the probability that each item wins does not drop too much.

\paragraph{Sample Complexity Lower Bound} 
For the sample complexity lower bound, we briefly describe the hard instance. One construction uses the following two single-dimensional distributions as building blocks.
\begin{align*}
    G = \begin{cases} 1 &  \text{with probability }~~ \frac{0.5}{n} \\ 0.5 & \text{with probability  }~~ \frac{1}{n} \\ 0 & \text{with probability  }~~ 1 - \frac{1.5}{n},
\end{cases}\ \ \ \ \ \ \ G^L=\begin{cases} 1 &  \text{with probability }~~ \frac{0.5 - \epsilon}{n} \\ 0.5 & \text{with probability  }~~ \frac{1 + \epsilon}{n} \\ 0 & \text{with probability  }~~ 1 - \frac{1.5}{n}.
\end{cases}~
\end{align*}
For a buyer with each item value drawn from $G$, the optimal item pricing sets the price of $qn$ items to be 0.5, and the rest of the items have price 1, where $q\approx \ln 4-1$. Suppose that the value distribution of the buyer is perturbed, such that the value distributions of $qn$ random items are switched to distribution $G^L$ that is dominated by $G$. Then the optimal item pricing will set the item prices of these dominated items to 0.5, and for each wrongly set item price, there is $\Omega(\frac{\eps}{n})$ loss in revenue. Thus, for a learning algorithm to learn a good item pricing mechanism with revenue loss $O(\eps)$, the learning algorithm needs to identify at least a constant fraction of the items with a perturbed value distribution. We then prove that the number of samples needed is at least $\Omega\left(\frac{n}{\eps^2}\right)$ by analyzing the Hellinger distance between the original value distribution and the perturbed value distribution.

\paragraph{Query Complexity Upper Bound} For the pricing query complexity upper bound of the unit-demand pricing problem, the first idea is to repeatedly query each distribution $D_i$ with price being all multiples of $\eps^2$ until the quantile at all discretized prices can be estimated accurately (with $\poly(n)/\eps^2$ queries). All item values can be rounded down to the closest multiple of $\eps^2$ with only $O(\eps)$ loss in revenue due to Nisan's $\eps$-IC to IC reduction (see Eg. \cite{BBHM-FOCS05,CHK-EC07}) for any single buyer. Therefore, $\poly(n)/\eps^4$ queries are sufficient. For further improvement, we need to run a non-uniform query algorithm after discretization. To be more specific, we prove that the accuracy we need for $D_i(v)$ is proportional to its quantile $1-F_{D_i}(v)$ for every $v\in[0,1]$. In other words, we need to make more pricing queries to an item distribution when the price is closer to 1, and fewer queries when the price is closer to 0. This way we are able to obtain a tighter dependency on $\eps$.

\paragraph{Query Complexity Lower Bound}

To prove the pricing query complexity lower bound for the ex-ante pricing problem, we use a similar idea as the sample complexity lower bound where we randomly perturb the base value distribution such that unless the learner can identify the way how most of the value distributions are perturbed, there could be a large revenue loss. Motivated by the $\Omega(\eps^{-3})$ query complexity lower bound example in the single-dimensional case \cite{leme2023pricing}, as the base distribution in our setting, we let each item have distribution $H$ with support $0,\frac{1}{2},\frac{1}{2}+\eps,\frac{1}{2}+2\eps,\cdots,\frac{3}{4}-\eps,\frac{3}{4}$ such that for every $k\in[\frac{1}{4\eps}]$, the seller obtains the same revenue $\frac{1}{2n}$ from $H$ by posting price $\frac{1}{2}+k\eps$. In other words, $H$ is an equal-revenue distribution with optimal revenue $\frac{1}{2n}$. Suppose that for each item $i$, we perturb its value distribution from $H$ as follows: select $k_i\in[\frac{1}{4\eps}-1]$ uniformly at random, and move the point mass of $H$ at $\frac{1}{2}+k_i\eps$ to $\frac{1}{2}+(k_i+1)\eps$. Then the optimal ex-ante pricing should have $p_i=\frac{1}{2}+(k_i+1)\eps$, and setting $p_i$ to any other discretized price leads to revenue loss $\Omega(\frac{\eps}{n})$. For each distribution, we can show that the query complexity to identify $k_i$ is $\Omega(\frac{n}{\eps^3})$, which means the overall query complexity lower bound for the ex-ante unit-demand pricing problem is $\Omega(\frac{n^2}{\eps^3})$. We conjecture that the lower bound example can be transformed to an $\Omega(\frac{n^2}{\eps^3})$ lower bound for the original unit-demand pricing problem.

\subsection{Discussion on the Multi-parameter Query Complexity Model}
While the economics and computation community has mostly agreed on a sample complexity model in the multi-dimensional setting, the definition of the query complexity model has never appeared in the literature. The most general query access model in a single round of interaction has three steps:
\begin{enumerate}
\item A fresh value vector $\bv\sim \bD$ is drawn from the value distribution.
\item The learner takes an action $a$ based on the past history.
\item A signal $\theta(\bv,a)$ is released to the learner.
\end{enumerate}
The \textit{sample access} model is a special case where the learner takes no action in each round, and the signal released to the learner is the sampled value vector $\bv$. In this paper, we are defining the \textit{pricing query complexity} to be the number of samples needed where the action in the second step is the \textit{pricing query access of a \textbf{single-dimensional distribution}} that has been well-motivated in the literature under different names. In fact, different action spaces and feedback signals in the query model correspond to different definitions of query complexity. For example, when we are allowed to simultaneously query $n$ item prices $(p_1,\cdots,p_n)$ and obtain a vector of feedback $(s_1,\cdots,s_n)$ where $s_i=\one[v_i\geq p_i]$, our query complexity upper bound and lower bound (for the ex-ante version) would be $\tilde{O}\left(\frac{n}{\eps^3}\right)$ and $\Omega\left(\frac{n}{\eps^3}\right)$ respectively. The general query access model has been studied extensively in the \textit{dynamic pricing} literature for regret minimization (Eg. \cite{kleinberg2003value}) under different action spaces and signaling schemes, but the query complexity problem should be equally exciting. It remains an interesting open direction what other well-motivated action and signaling schemes lead to non-trivial query complexity results.

\subsection{Related Work}
\label{subsec:related}

\paragraph{Multi-dimensional revenue maximization} There has been a growing literature on multi-parameter Bayesian revenue maximization during the last two decades. Due to the complexity of the optimal mechanisms, most work has been focused on using simple mechanisms (such as item pricing, grand bundle pricing, sequential posted pricing, two-part tariff) to approximate the revenue obtained by optimal mechanisms for buyers with special value structures such as subadditivity, with a non-exhaustive list including \cite{CHK-EC07,chawla2010multi,CMS-GEB15,babaioff2020simple,yao2014n,rubinstein2018simple,chawla2016mechanism,CDW-SICOMP21,correa2023constant,cai2023simultaneous}. The computation of the optimal multi-dimensional pricing for a unit-demand buyer has been shown to be NP-hard even with bounded support size or i.i.d. settings \cite{daskalakis2012optimal,chen2018complexity}. If we want an approximately optimal item pricing, \cite{cai2015extreme} proposes an algorithm that obtains an additive Polynomial Time Approximation Scheme, and the existence of a multiplicative PTAS remains an open question.

\paragraph{Sample complexity} As a natural generalization of Bayesian revenue maximization, revenue maximization with samples has gained increasing popularity since \cite{cole2014sample}. In single-parameter settings, literature including \cite{elkind2007designing, dhangwatnotai2010revenue, huang2015making, morgenstern2015pseudo, DHP-STOC16, hartline2019sample, gonczarowski2017efficient, guo2019settling, jin2023learning} studies how to learn the optimal (or approximately optimal) auctions via samples, with \cite{guo2019settling} getting the tight sample complexity of the optimal auctions for most important distribution classes.
In multi-dimensional settings, results are mostly obtained to prove the existence of an upper bound on the sample complexity for obtaining near-optimal mechanisms for independent (or slightly correlated) items \cite{dughmi2014sampling,balcan2016sample,morgenstern2016learning,cai2017learning,syrgkanis2017sample,balcan2023generalization,gonczarowski2021sample,brustle2020multi,GHTZ-COLT21,balcan2021much,cai2022computing,CHMY-EC23,jin2024sample}. The sample complexity result for optimal mechanisms cannot go beyond independent items since for complicated correlated value distributions, slightly perturbing the values may lead to drastic change in optimal revenue \cite{psomas2019smoothed, chawla2020menu}.

\paragraph{Pricing queries} The usage of pricing queries for revenue maximization can date back to \cite{kleinberg2003value} which studies the regret of online pricing with different distributional assumptions. While most work in the literature has been about dynamic pricing for bandit regret minimization, recently several papers started to study using these threshold queries to learn the properties of a distribution, such as estimating the mean, median, optimal monopoly price, and the entire CDF \cite{meister2021learning,leme2023pricing,okoroafor2023non,paes2023description, SW-EC24}. \cite{leme2023pricing} first studies the query complexity of the threshold query model in the single-dimensional setting, and we first use the same query model to study the query complexity in the multi-dimensional setting.

\paragraph{Revenue (non-)monotonicity}  A considerable body of work explores revenue monotonicity and non-monotonicity across various mechanisms \cite{hart2015maximal, devanur2016sample, rubinstein2018simple, Yao-SAGT18, guo2019settling, GHTZ-COLT21, CHMY-EC23, cai2023simultaneous}. Our paper's concept of approximate strong revenue monotonicity is most closely related to the "revenue Lipschitzness" variant introduced in \cite{CHMY-EC23}, which bounds revenue change by the distance between original and empirical distributions. Separately, a similar "approximate revenue monotonicity" concept, discussed in \cite{rubinstein2018simple, Yao-SAGT18, cai2023simultaneous}, establishes a different form of monotonicity for a broader class of buyers: if distribution $\bD$ stochastically dominates $\bE$, the optimal revenue under $\bD$ is guaranteed to be at least a constant fraction of the optimal revenue under $\bE$. Crucially, while the latter concept uses a multiplicative approximation, our paper's "approximate strongly monotonicity" provides an additive approximation, offering a distinct perspective.

\subsection{Paper Organization}

In \Cref{sec:prelim} we define the notations and useful concepts in our paper. In \Cref{sec:sample-upper} we give a complete proof of the $\tilde{O}(n/\eps^2)$ sample complexity upper bound of \SUDPP, with the proof of some claims in the section deferred to \Cref{sec:appendix-sample-upper}. In \Cref{sec:query-upper} we prove the $\tilde{O}(n^2/\eps^3)$ query complexity upper bound of \SUDPP, with the proof of some lemmas deferred to \Cref{sec:appendix-query-upper}. The $\Omega(n/\eps^2)$ sample complexity lower bound of \SUDPP is deferred to \Cref{sec:sample-lower}. The $\Omega(n^2/\eps^3)$ query complexity lower bound of the ex-ante version of \SUDPP is deferred to \Cref{sec:query-lower}.

\section{Preliminaries}
\label{sec:prelim}

\subsection{Notations}

We study a Bayesian multi-parameter revenue maximization problem with $n$ items and a single buyer. The buyer has independent item values with support $[0,1]$, which means that the buyer's value distribution $\bD=D_1\times D_2 \times\cdots\times D_n$ is a product distribution on $[0,1]^n$. For every single-dimensional distribution $D$, we use $F_D(\cdot)$ and $f_D(\cdot)$ to denote the CDF and the PDF of the distribution. 

For any value distribution $\bD$ and item pricing $\bp$, denote by $\rev_{\bD}(\bp)$ the expected revenue of $\bp$ when the buyer's value vector is drawn from $\bD$. We let 
\[\rev_{\bD}^*=\max_{\bp}\rev_{\bD}(\bp)\] 
denote the optimal revenue obtained by any item pricing, and 
\[\bp_{\bD}^*=\mathop{\arg\max}_{\bp}\rev_{\bD}(\bp)\]
denote the optimal posted prices. Consistent with the notation in the literature (Eg. \cite{CHK-EC07,chen2018complexity}), we use the abbreviation \SUDPP to refer to the Bayesian Unit-demand Pricing Problem that computes the revenue-optimal item pricing $\bp_{\bD}^*$ for any buyer distribution $\bD$.

While we assume that the learner does not have full knowledge about $\bD$, the learner has sample or (pricing) query access to the distributions.  The goal of the learner is to learn a near-optimal item pricing $\bp=(p_1,p_2,\cdots,p_n)$ with revenue $\rev_{\bD}(\bp)$ at least $\rev_{\bD}^*-\eps$ with probability at least $1-\delta$. 

When we study the sample complexity model, with $N$ sampled valuation vectors $\bv_1,\bv_2,\cdots,\bv_N$ we will learn an empirical product distribution $\bE=E_1\times E_2\times \cdots\times E_n$ where $E_i$ is the uniform distribution over $N$ sampled values of the $i$th coordinate $v_{1,i},v_{2,i},\cdots,v_{N,i}$.

\subsection{Query Complexity and Computational Complexity}

The sample complexity of our learning problem is the minimum number of samples $N$ such that there exists a learning algorithm that can take $N$ samples of any distribution $\bD$ to learn a near-optimal item pricing with revenue loss at most $\eps$ with probability $1-\delta$. When we study the query complexity, the learning algorithm may take adaptive queries. Notice that we do not assume that the learning algorithm needs to be efficient with running time independent of $\eps$. In fact, it is known that even when the buyer's value distribution is exactly known to the seller, computing the optimal item pricing vector $\bp=(p_1,\cdots,p_n)$ that maximizes the expected revenue is already NP-hard \cite{chen2018complexity}. In our learning algorithm, we will use an additive PTAS algorithm by \cite{cai2015extreme} as a subroutine to solve the unit-demand pricing problem with $\eps$ loss. This means that our learning algorithm will also be efficient for any fixed $\eps>0$.

\begin{Theorem}[Theorem 1 of \cite{cai2015extreme}]
\label{lma:ptas}
    For \SUDPP~problem with value distribution $\bD$, there exists an algorithm that runs in $O\Big(n^{\frac{\log^3 \epsilon^{-1}}{\epsilon^4}} \Big)$ time and outputs $\bp$, such that $\rev_{\bD}(\bp') \geq \rev^*_{\bD} - \eps$ for any $\eps>0$.
\end{Theorem}

\subsection{Ex-ante Unit-demand Pricing}
We define an \textit{ex-ante} version of the unit-demand pricing problem as follows. For any product value distribution $\bD=D_1\times D_2\times\cdots\times D_n$ and item pricing $\bp=(p_1,p_2,\cdots,p_n)$, define the \textit{ex-ante revenue} of $\bp$ to be
\begin{eqnarray}
\erev_{\bD}(\bp)&:=&\max_{q_1,q_2,\cdots,q_n} \sum_{i=1}^{n}q_ip_i \label{prog:exante}\\
&s.t.& \sum_{i=1}^{n}q_i\leq 1; \notag\\
& & q_i\leq \Pr{v_i\geq p_i},\ \forall i\in[n]. \notag
\end{eqnarray}
In other words, for any price vector, the agent will purchase as many affordable items as possible, such that one item is sold on average. Intuitively this is a simpler problem compared to the original problem as it does not take the buyer's utility ordering of all items into account. We are going to provide a lower bound on the query complexity of this problem.

\subsection{Distribution Distance Metrics}
We need to analyze the distance between different distributions in the analysis of both upper bounds and lower bounds. In the upper bound proofs, we prove that when the empirical distribution and the original distribution are \textit{close}, their revenue under any pricing we specify is also close. The distance metric we use here is the \textbf{Kolmogorov\ distance} (although we do not specifically call its name in the proof). For any two single-dimensional distributions $D$ and $E$,
\[d_{Kolmogorov}(D,E)=\sup_{v}|F_D(v)-F_E(v)|.\]
In the lower bound proofs, we use the fact that when two distributions are ``\textit{close}'', it is hard to distinguish the two distributions with a small number of samples. The distance metric here is the \textbf{Total Variation (TV) \ distance}. For any measurable space $\Omega,\mathcal{F}$ and probability measures $P,Q$, the total variation distance is defined by
\[d_{TV}(P,Q)=\sup_{A\in \mathcal{F}}|P(A)-Q(A)|.\]
When two distributions have TV distance $\delta$, suppose that the learner is given a sample from one of the distributions but is not told the origin of the sample. The probability that the learner can correctly guess the origin of the sample is at most $\frac{1}{2}+\delta$. This means that when the TV distance between two distributions is small, it is hard to distinguish the two distributions with a single sample.

The multi-dimensional TV distance is hard to analyze. Therefore, we use the \textbf{Hellinger \ distance} to bound the TV distance in the multi-dimensional setting. The (square of the) Hellinger distance is defined by
\[d_H^2(P,Q)=\frac{1}{2}\int_{\Omega}\left(\sqrt{P(dx)}-\sqrt{Q(dx)}\right)^2.\]
A nice property of Hellinger distance (see Eg. \cite{gibbs2002choosing} for the reference of the properties of the Hellinger distance) is that when $(\Omega,\mathcal{F})$ is a product space, 
\[1-d_H^2(P,Q)=1-\prod_{i}d_H^2(P_i,Q_i).\]
The TV distance is bounded by the Hellinger distance via the following inequality:
\[d_{TV}(P,Q)\leq \sqrt{2}d_H(P,Q).\]

\subsection{Concentration Inequalities}
To prove that two distributions are close (under Kolmogorov distance), in both the sample complexity problem and the query complexity problem we need to estimate the quantile of an empirical distribution at a specific value. The major concentration inequality we use is Bernstein's Inequality, which provides a better bound compared to the Chernoff-Hoeffding bound when the quantile is close to 0 or 1. This is especially helpful for us in getting a tighter bound for the query complexity.

\begin{Theorem}[Bernstein's Inequality for Bounded Variables]
\label{thm:Bernstein-bounded}
Let $X_1, \ldots, X_N$ be independent mean-zero random variables such that $|X_i| \leq M$ for all $i$ and $\sigma^2 := \sum_{i \in [N]} \E[X^2_i]$. Then, for any $\varepsilon \geq 0$, 
\[
\pr \left[\sum_{i = 1}^N X_i \geq \varepsilon \right] \leq \exp\left(- \frac{\varepsilon^2/2}{\sigma^2 + M\varepsilon/3}\right) \quad \text{and} \quad \pr \left[\sum_{i = 1}^N X_i \leq -\varepsilon \right]\leq \exp\left(- \frac{\varepsilon^2/2}{\sigma^2 + M\varepsilon/3}\right) \enspace .
\]
\end{Theorem}

Together with the union bound, \cite{GHTZ-COLT21} derived a multi-dimensional version of Bernstein's Inequality which will be helpful in the proof of the sample complexity upper bound.

\begin{Lemma}[Lemma 5 of \cite{GHTZ-COLT21}]
\label{lma:Bernstein}
    For any value distribution $\bD$ with support $[0,1]^n$, $N$ and $\delta \in (0, 1)$, consider the empirical distribution $\bE$ from $N$ i.i.d. samples from $\bD$. Then, with probability at least $1 - \delta$, we have
    \[
    |F_{D_i}(v) - F_{E_i}(v)| ~\leq~ \sqrt{F_{D_i}(v) (1 - F_{D_i}(v)) \cdot \frac{2 \log (2nN \delta^{-1})}{N}} ~+~ \frac{\log (2nN \delta^{-1})}{N}
    \]
    for all $i \in [n]$ and $v \in [0, 1]$.
\end{Lemma}

For simplicity of the notation, we define 
\[
\Gamma ~:=~ \frac{\log (2n N \delta^{-1})}{N}
\]
to be the parameter for the learning error, which is further parameterized by $n, \delta, N$. Then, the bound in \Cref{lma:Bernstein} can be simplified as 
\[
|F_{D_i}(v) - F_{E_i}(v)| ~\leq~ \sqrt{F_{D_i}(v) (1 - F_{D_i}(v)) \cdot 2\Gamma} ~+~ \Gamma.
\]

\section{$n/\epsilon^2$ Upper Bound for Sample Complexity}
\label{sec:sample-upper}

In this section, we prove the sample complexity of \SUDPP problem has a $\OTild(n/\epsilon^2)$ upper bound.
To prove the upper bound, we adopt the idea of strong revenue monotonicity. Introduced by \cite{DHP-STOC16}, a stochastic revenue maximization problem is called \textit{strongly revenue monotone} if the revenue does not decrease under the same auction setting, when we replace the underlying value distribution with another distribution that strongly dominates it. \cite{GHTZ-COLT21} shows a $\OTild(\frac{n}{\epsilon^2})$ sample complexity upper bound for all strongly monotone stochastic optimization problems. However, as the strong monotonicity does not hold for our problem,  we prove a weaker ``\textit{approximately strong monotonicity}'', which was introduced by \cite{CHMY-EC23} and shown to be powerful for proving sample complexity upper bounds for single-parameter auctions. Our contribution is to prove the approximately strong monotonicity for \SUDPP, and further show that the approximately strong monotonicity guarantees the same sample complexity upper bound for strongly monotone problems. Specifically, we prove the following sample complexity:

\begin{Theorem}
\label{thm:sample-upper}
    For \SUDPP problem with $\bD$ being the product value distribution, there exists an algorithm that uses $N = O( \log^4(n/(\epsilon\delta)) \cdot \log^2(\log(n/(\epsilon\delta))) \cdot n \cdot \epsilon^{-2})$ samples from $\bD$ and outputs $\bp = (p_1, \cdots, p_n)$, such that $\rev_{\bD}(\bp) \geq \rev^*_{\bD}  - \epsilon$ holds with probability $1 - \delta$. 
    Furthermore, the algorithm runs in $O(n^{\poly(\epsilon^{-1})})$ time.
\end{Theorem}

\subsection{Proof Sketch for \Cref{thm:sample-upper}}

Our main idea is to show the following \Cref{alg:sample} is the desired algorithm for \Cref{thm:sample-upper}. 

\begin{algorithm}
\caption{\textsc{Learning $\bD$ via Samples}}
\label{alg:sample}
\begin{algorithmic}[1]
\State \textbf{input:} \SUDPP instance with value distribution $\bD$, error bound $\epsilon$, failure probability bound $\delta$.
\State Reveal $N = C \cdot \log^4(n/(\epsilon\delta)) \cdot \log^2(\log(n/(\epsilon\delta))) \cdot n \cdot \epsilon^{-2}$ i.i.d. samples from the product distribution $\bD$, and construct the empirical distribution $\bE$.
\State Run an Additive-PTAS algorithm for $\bE$, and get price vector $\bp' = (p'_1, \cdots, p'_n)$, such that $\rev_{\bE}(\bp') \geq \rev^*_{\bE} - \frac{\epsilon}{3}$.
\State \textbf{output:} price vector $\bp'$.
\end{algorithmic}
\end{algorithm}

The proofs go via the following four steps:

\noindent \textbf{Step 1: The empirical distribution concentrates to $\bD$.} This is shown via \Cref{lma:Bernstein}.

\noindent \textbf{Step 2: When $\bE$ is close to $\bD$,  $\rev^*_{\bE}$ is close to $\rev^*_{\bD}$.} The second step is to show that we still gain a good revenue guarantee when switching from the original distribution $\bD$ to the empirical distribution $\bE$. To achieve this, we prove the following stronger \Cref{lma:rev-close}:

\begin{Lemma}
\label{lma:rev-close}
    For any product value distribution $\bD$ and $\bE$ with support $[0,1]^n$, suppose condition $|F_{D_i}(v) - F_{E_i}(v)| ~\leq~ \sqrt{F_{D_i}(v) (1 - F_{D_i}(v)) \cdot 2\Gamma} ~+~ \Gamma$ holds for every $i \in [n]$ and $v \in [0, 1]$. Then, for any price vector $\bp \in [0,1]^n$, we have
    \[
    \rev_{\bE}(\bp) ~\geq~ \rev_{\bD}(\bp) - 300\log \Gamma^{-1} \cdot \left(\Gamma n + \sqrt{\log \Gamma^{-1} \cdot \Gamma n}\right).
    \]
\end{Lemma}

When \Cref{lma:rev-close} is satisfied, taking $\bp= \bp^*_{\bD}$ implies $\rev^*_{\bE}$ is close to $\rev^*_{\bD}$. We defer the proof of \Cref{lma:rev-close} to \Cref{sec:main-lma}.

\noindent \textbf{Step 3: There is a good PTAS algorithm for $\bE$.} The third step is to run a PTAS algorithm for the empirical distribution $\bE$. We apply the additive PTAS algorithm in \Cref{lma:ptas} by \cite{cai2015extreme} on the empirical distribution $\bE$ to get a pricing $\bp'$ such that
\[\rev_{\bE}(\bp')\geq \rev_{\bE}^*-\frac{\eps}{2}.\]

\noindent \textbf{Step 4: $\bp'$ is also near-optimal for $\bD$.} The last step of our proof is to show that the price vector $\bp'$ given by \Cref{lma:ptas} is also near-optimal for distribution $\bD$. Our main idea is to apply \Cref{lma:rev-close} and show that $\rev_{\bD}(\bp')$ is close to $\rev_{\bE}(\bp')$. This is achievable when we convert the CDF bound given by \Cref{lma:Bernstein} to a function of $F_{E_i}(v)$: 

\begin{restatable}{Claim}{dtoe}
    \label{clm:cdf-dtoe}
    For any product value distribution $\bD$ and $\bE$ with support $[0,1]^n$, suppose condition $|F_{D_i}(v) - F_{E_i}(v)| ~\leq~ \sqrt{F_{D_i}(v) (1 - F_{D_i}(v)) \cdot 2\Gamma} + \Gamma$ holds for every $i \in [n]$ and $v \in [0, 1]$. Then, for every $i \in [n]$ and $v \in [0, 1]$ for some $\Gamma \leq 0.01$, we have
    \[
    |F_{D_i}(v) - F_{E_i}(v)| ~\leq~ \sqrt{F_{E_i}(v) (1 - F_{E_i}(v)) \cdot 16 \Gamma} ~+~ 8 \Gamma.
    \]
\end{restatable}

We defer the proof of \Cref{clm:cdf-dtoe} to \Cref{sec:dtoe}, and first finish the proof of \Cref{thm:sample-upper}:

\begin{proof}[Proof of \Cref{thm:sample-upper}]
    Recall that we set $N = C \cdot \log^4(n/(\epsilon\delta)) \cdot \log^2(\log(n/(\epsilon\delta))) \cdot n \cdot \epsilon^{-2}$ for a sufficiently large constant $C$,  and defined $\Gamma = \log(2nN\delta^{-1}) \cdot N^{-1}$. Then, we have 
    \begin{align}
        \Gamma = \frac{\epsilon^2}{ C' \cdot n \cdot \log^3(n/(\epsilon\delta)) \cdot \log(\log(n/(\epsilon\delta)))} \label{eq:gamma-exact}
    \end{align}
    for some sufficiently large constant $C'$. We assume in the proof that $\Gamma < 0.01$. This is true when $C$ and the corresponding $C'$ are sufficiently large.

    To prove \Cref{thm:sample-upper}, we first apply \Cref{lma:Bernstein}, which guarantees that with probability $1 - \delta$ we have $|F_{D_i}(v) - F_{E_i}(v)| ~\leq~ \sqrt{F_{D_i}(v) (1 - F_{D_i}(v)) \cdot 2\Gamma} ~+~ \Gamma$ for every $i \in [n]$ and $v \in [0, 1]$. Assuming  $|F_{D_i}(v) - F_{E_i}(v)| ~\leq~ \sqrt{F_{D_i}(v) (1 - F_{D_i}(v)) \cdot 2\Gamma} ~+~ \Gamma$ holds, applying \Cref{lma:rev-close} with $\bp = \bp^*_{\bD}$ gives
\begin{align}
    \rev^*_{\bE} ~\geq~ \rev_{\bE}(\bp^*_{\bD}) ~&\geq~  \rev_{\bD}(\bp^*_{\bD}) - 300\log \Gamma^{-1} \cdot \left(\Gamma n + \sqrt{\log \Gamma^{-1} \cdot \Gamma n}\right) \notag \\
    ~&\geq~ \rev^*_{\bD} - 300\log \Gamma^{-1} \cdot \left(\Gamma n + \sqrt{\log \Gamma^{-1} \cdot \Gamma n}\right) \label{eq:revd-to-reve}.
\end{align}

Next, we apply the PTAS algorithm given by \Cref{lma:ptas}, which guarantees that the output vector $\bp'$ satisfies
\begin{align}
\label{eq:reve-to-bpprime}
    \rev_{\bE}(\bp') ~\geq~ \rev^*_{\bE} - \frac{\epsilon}{2}
\end{align}

Finally, as \Cref{clm:cdf-dtoe} guarantees that $|F_{D_i}(v) - F_{E_i}(v)| ~\leq~ \sqrt{F_{E_i}(v) (1 - F_{E_i}(v)) \cdot 16 \Gamma}$, where the required condition $\Gamma \leq 0.01$ for \Cref{clm:cdf-dtoe} can be satisfied when $C'$ is sufficiently large, we can reuse \Cref{lma:rev-close} with $\Gamma' = 8\Gamma$ and $\bp = \bp'$, which gives
\begin{align}
\label{eq:bpprime-to-revd}
    \rev_{\bD}(\bp') ~\geq~ \rev_{\bE}(\bp') - 2400\log \Gamma^{-1} \cdot \left(\Gamma n + \sqrt{\log \Gamma^{-1} \cdot \Gamma n}\right).
\end{align}
Summing \eqref{eq:revd-to-reve}, \eqref{eq:reve-to-bpprime}, \eqref{eq:bpprime-to-revd} together, we get
\[
\rev_{\bD}(\bp') ~\geq~ \rev^*_{\bD} - \frac{\epsilon}{3} - 3200\log \Gamma^{-1} \cdot \left(\Gamma n + \sqrt{\log \Gamma^{-1} \cdot \Gamma n}\right).
\]
To prove \Cref{thm:sample-upper}, it remains to show $3200\log \Gamma^{-1} \cdot \left(\Gamma n + \sqrt{\log \Gamma^{-1} \cdot \Gamma n}\right) \leq \frac{\epsilon}{2}$. Plugging \eqref{eq:gamma-exact} into the calculation, we have 
\begin{align*}
    \log \Gamma^{-1} \cdot \left(\Gamma n + \sqrt{\log \Gamma^{-1} \cdot \Gamma n}\right) ~&\leq~ 2\sqrt{\log^3 \Gamma^{-1} \cdot \Gamma n} \\
    ~&=~ 2\epsilon \cdot \sqrt{\frac{\log^3 \Gamma^{-1}}{C' \cdot \log^3(n/(\epsilon\delta)) \cdot \log(\log(n/(\epsilon\delta)))}} \\
    ~&\leq~ 2\epsilon \cdot \sqrt{\frac{(8 \log C' \cdot \log (n/(\epsilon \delta)))^3}{C' \cdot \log^3(n/(\epsilon\delta)) \cdot \log(\log(n/(\epsilon\delta)))}} \\
    ~&\leq~ 100 \epsilon \cdot \sqrt{\frac{\log^3 C'}{C'}} ~\leq~ \frac{1}{3200} \cdot \frac{\epsilon}{2},
\end{align*}
where the last inequality holds when $C'$ is sufficiently large.
\end{proof}

\subsection{Proof of \Cref{lma:rev-close}}
\label{sec:main-lma}
In this subsection, we prove \Cref{lma:rev-close}. We first define product distribution $\bbD$: For every $i \in [n]$ and $v \in [0, 1]$, we define
\begin{align}
\label{eq:dhat-def}
    F_{\breve D_i}(v) ~:=~ \max\left\{1, F_{D_i}(v) +  \sqrt{F_{D_i}(v) (1 - F_{D_i}(v)) \cdot 2\Gamma} ~+~ \Gamma\right\}.
\end{align}
Then, the following lemma guarantees that the total variation distance between $\bbD$ and $\bD$ is small:

\begin{Lemma}[Lemma 2 and Lemma 11 in \cite{GHTZ-COLT21}]
\label{lma:tv-d-dhat}
    Let $\bD$ be a product distribution with support $[0, 1]^n$ and $\bbD$ be the product distribution defined by \eqref{eq:dhat-def}. The total variation distance between $\bD$ and $\bbD$ is bounded by $4\sqrt{n\Gamma \cdot \log \Gamma^{-1}}$. 
\end{Lemma}

\Cref{lma:tv-d-dhat} guarantees that $\rev_{\bbD}(\bp)$ and $\rev_{\bD}(\bp)$ are sufficiently close. It remains to show that $\rev_{\bE}(\bp)$ is close to $\rev_{\bbD}(\bp)$.

The following claim first bounds the difference between $F_{E_i}(v)$ and $F_{\breve D_i}(v)$ as a function of $F_{\breve D_i}(v)$:

\begin{restatable}{Claim}{clmcdfboundbbd}
\label{clm:cdf-bound-bbd}
    For product distributions $\bD$ and $\bE$ with support $[0, 1]^n$, assume for every $i \in [n]$ and $v \in [0, 1]$ we have $|F_{D_i}(v) - F_{E_i}(v)| ~\leq~ \sqrt{F_{D_i}(v) (1 - F_{D_i}(v)) \cdot 2\Gamma} ~+~ \Gamma$. Let $\bbD$ be the product distribution defined by \eqref{eq:dhat-def}. Then, we have
    \[
    0 ~\leq~ F_{\breve D_i}(v) - F_{E_i}(v) ~\leq~ 4\sqrt{(1 - F_{\breve D_i}(v)) \cdot \Gamma} + 12\Gamma
    \]
    for every $i \in [n]$ and $v \in [0, 1]$.
\end{restatable}

We defer the proof of \Cref{clm:cdf-bound-bbd} to \Cref{sec:clmcdfboundbbd}. With \Cref{clm:cdf-bound-bbd}, the final piece for proving \Cref{lma:rev-close} is the approximated strongly monotone property of \SUDPP. To be specific, the following lemma suggests that for two product distributions $\bG$ and $\bH$ that satisfy $\bG \preceq \bH$, if the CDFs of $\bG$ and $\bH$ are sufficiently close, moving from $\bG$ to $\bH$ does not lose too much revenue:

\begin{Lemma}
    \label{lma:approx-sm}
    Let $\bG$ and $\bH$ be two product distributions with support $[0, 1]^n$. Assume for every $i \in [n]$ and $v \in [0, 1]$, we have
    \[
    0 ~\leq~ F_{G_i}(v) - F_{H_i}(v) ~\leq~ \sqrt{(1 - F_{G_i}(v)) \cdot \gamma} + \gamma.
    \]
     Then, for any price vector $\bp$, we have 
    \[
    \rev_{\bH}(\bp) \geq \rev_{\bG}(\bp) - 7\log \gamma^{-1} \cdot \left(\gamma n + \sqrt{\log \gamma^{-1} \cdot \gamma n}\right).
    \]
\end{Lemma}

We defer the proof of \Cref{lma:approx-sm} to \Cref{sec:approx-sm}, and first finish the proof of \Cref{lma:rev-close}:
\begin{proof}[Proof of \Cref{lma:rev-close}]
    Consider to define product distribution $\bbD$ via \eqref{eq:dhat-def}. Then, \Cref{lma:tv-d-dhat} guarantees that the total variation distance between $\bD$ and $\bbD$ is bounded by $4\sqrt{n\Gamma \cdot \log \Gamma^{-1}}$. As a corollary, we have $\rev_{\bbD}(\bp) \geq \rev_{\bD}(\bp) - 4\sqrt{n\Gamma \cdot \log \Gamma^{-1}}$. 

    Next, since \Cref{clm:cdf-bound-bbd} guarantees that \[
    0 ~\leq~ F_{\breve D_i}(v) - F_{E_i}(v) ~\leq~ \sqrt{(1 - F_{\breve D_i}(v)) \cdot 16\Gamma} + 12\Gamma
    \]
    for every $i \in [n]$ and $v \in [0, 1]$, applying \Cref{lma:approx-sm} with $\bG = \bbD$, $\bH = \bE$, and $\gamma = 16\Gamma$ gives 
    \begin{align*}
        \rev_{\bE}(\bp) ~&\geq~ \rev_{\bbD}(\bp) - 7\log (16\Gamma)^{-1} \cdot \left(16\Gamma n + \sqrt{\log (16\Gamma)^{-1} \cdot 16\Gamma n}\right) \\
        ~&\geq~ \rev_{\bbD}(\bp) - 200\log \Gamma^{-1} \cdot \left(\Gamma n + \sqrt{\log \Gamma^{-1} \cdot \Gamma n}\right).
    \end{align*}
    Combining it with the inequality that $\rev_{\bbD}(\bp) \geq \rev_{\bD}(\bp) - 4\sqrt{n\Gamma \cdot \log \Gamma^{-1}}$ gives
    \[
    \rev_{\bE}(\bp) ~\geq~ \rev_{\bD}(\bp) - 300\log \Gamma^{-1} \cdot \left(\Gamma n + \sqrt{\log \Gamma^{-1} \cdot \Gamma n}\right). \qedhere
    \]
\end{proof}

\subsection{\SUDPP is Approximately Strongly Monotone: Proof of \Cref{lma:approx-sm}}
\label{sec:approx-sm}

\paragraph{Extra Notations.} We first define some extra notations needed in the proof of \Cref{lma:approx-sm}.

For $i \in [n]$ and $v \in [0, 1]$, let $f_{G_i}(v)$ and $f_{H_i}(v)$ be the PDF of distribution $G_i$ and $H_i$, respectively. For simplicity of notations, we generalize the definition of PDF to discrete distributions: For instance, if there is a point mass at value $v$ for distribution $G_i$, then $f_{G_i}(v)$ is defined as $\delta(0) \cdot \pr_{X \sim G_i}[X = v]$, where $\delta$ represents Dirac delta function. 

To further handle the point masses in $\bG$ and $\bH$, we define 
\[
F_{G_i}(v^-)~:=~ \lim_{x \to v^-} F_{G_i}(x) ~=~ \mathop{\pr}\limits_{X \sim G_i}[X < v]  \quad \text{and} \quad F_{H_i}(v^-)~:=~ \lim_{x \to v^-} F_{H_i}(x) ~=~ \mathop{\pr}\limits_{X \sim H_i}[X < v].
\]
Note that $F_{G_i}(v)$ represents the probability that $X \leq v$ for $X \sim G_i$, so $F_{G_i}(v) = F_{G_i}(v^-) + \pr_{X \sim G_i}[X = v]$, and similarly $F_{H_i}(v) = F_{H_i}(v^-) + \pr_{X \sim H_i}[X = v]$.

Without loss of generality, we assume $p_1 \leq p_2 \leq \cdots \leq p_n$, where $\bp = (p_1, \cdots, p_n)$ is the prices vector stated in \Cref{lma:approx-sm}. For $i \in [n]$ and $\theta \in [0, 1 - p_i]$, we define 
\[
Q_{G_i}(\theta) ~:=~ \prod_{j < i} F_{G_j}(p_j + \theta) \cdot \prod_{j > i} F_{G_j}((p_j + \theta)^-)\quad \text{and} \quad Q_{H_i}(\theta) ~:=~ \prod_{j < i} F_{H_j}(p_j + \theta) \cdot \prod_{j > i} F_{H_j}((p_j + \theta)^-).
\]
The notation $Q_{G_i}(\theta)$ and $Q_{H_i}(\theta)$ represent the probability that item $i$ wins the auction (assuming the underlying product distributions are $\bG$ and $\bH$, respectively) when the value of item $i$ is $p_i + \theta$. Note that $Q_{G_i}(\theta)$ and $Q_{H_i}(\theta)$ are also a function of $\bp$. We omit $\bp$ in the notations for simplicity. In both notations, $\theta$ represents the utility of the buyer. Note that we break the ties by asking the buyer to take the item with the highest price. Therefore, in both notations, we count the probability that $v_j \leq p_j +\theta$ for $j < i$, and for $j > i$, we count the probability that $v_j < p_j +\theta$.

With notations $Q_{G_i}(\theta)$ and $Q_{H_i}(\theta)$, we define
\[
P_{G_i} := \int_0^{1 - p_i} f_{G_i}(p_i + \theta) \cdot Q_{G_i}(\theta) d\theta \quad \text{and} \quad P_{H_i} ~:=~ \int_0^{1 - p_i} f_{H_i}(p_i + \theta) \cdot Q_{H_i}(\theta) d\theta
\]
to be the probability that item $i$ is the agent's favorite item when the underlying product distributions are $\bG$ and $\bH$, respectively. Note that $P_{G_i}$ and $P_{H_i}$ are also a function of $\bp$. We omit $\bp$ in the notations for simplicity.

Finally, we define
\[
S_{G_i}(\theta) ~:=~ \sum_{j \leq i} \left(1 - F_{G_j}(p_j + \theta)\right) + \sum_{j > i} \left(1 - F_{G_j}((p_j + \theta)^-)\right).
\]
Notation $S_{G_i}(\theta)$ approximately represents the total probability mass that prevents $i$ from winning the auction with value $\theta + p_i$. It is off by an extra $1 - F_{G_i}(p_i + \theta)$ term. We include this term in $S_{G_i}(\theta)$, as it immediately gives the following monotone properties:
\begin{itemize}
    \item For a fixed $i$ and $\theta < \theta'$, we have $S_{G_i}(\theta) \geq S_{G_i}(\theta')$. This follows from the monotonicity of CDFs.
    \item For a fixed $\theta$ and $i < j$, we have $S_{G_i}(\theta) \geq S_{G_j}(\theta)$. This follows from the observation that $S_{G_i}(\theta) - S_{G_{i+1}}(\theta) = \pr_{X \sim D_i}[X = p_i + \theta]$.
\end{itemize}

\paragraph{Decomposing $\rev_{\bG}(\bp) - \rev_{\bH}(\bp)$.} Note that
\[
\rev_{\bG}(\bp) ~=~ \sum_{i \in [n]} p_i \cdot P_{G_i} \quad \text{and} \quad \rev_{\bH}(\bp) ~=~ \sum_{i \in [n]} p_i \cdot P_{H_i}.
\]
We bound the difference between $\rev_{\bG}(\bp)$ and $\rev_{\bH}(\bp)$  by
\begin{align}
\label{eq:rev-to-prob}
    \rev_{\bG}(\bp) - \rev_{\bH}(\bp) ~=~ \sum_{i \in [n]} p_i \cdot (P_{G_i} - P_{H_i}) ~\leq~ \sum_{i \in [n]}   \max\{0, P_{G_i} - P_{H_i} \},
\end{align}
where the last inequality follows from the fact that $p_i \leq 1$, and the non-negativity of $\max\{0, P_{G_i} - P_{H_i} \}$. Therefore, to bound $\rev_{\bG}(\bp) - \rev_{\bH}(\bp)$, it's sufficient to give a non-negative upper bound for every $P_{G_i} - P_{H_i}$.

We start from decomposing a single $P_{G_i} - P_{H_i}$ into two parts: We have
\begin{align}
    P_{G_i} - P_{H_i} ~=&~ \int_0^{1 - p_i} f_{G_i}(p_i + \theta) \cdot Q_{G_i}(\theta) d\theta - \int_0^{1 - p_i} f_{H_i}(p_i + \theta) \cdot Q_{H_i}(\theta) d\theta \notag \\
    ~=&~ \int_0^{1 - p_i} \left(f_{G_i}(p_i + \theta) - f_{H_i}(p_i + \theta)\right) \cdot Q_{H_i}(\theta) d\theta \label{eq:part1} \\
    &+ \int_0^{1 - p_i} f_{G_i}(p_i + \theta) \cdot \left( Q_{G_i}(\theta) -Q_{H_i}(\theta)\right)  d\theta \label{eq:part2}
\end{align}

\eqref{eq:part1} can be further bounded as 
\begin{align*}
    \text{\eqref{eq:part1}}~=&~ \Big(\big(F_{G_i}(p_i + \theta)- F_{H_i}(p_i + \theta)\big) \cdot Q_{H_i}(\theta)\Big) \Big|_{0}^{1 - p_i} \\
    &-~ \int_{0}^{1 - p_i} (F_{G_i}(p_i + \theta)- F_{H_i}(p_i + \theta)) \cdot Q'_{H_i}(\theta) d\theta \\
    ~=&~ \big(F_{G_i}(1)- F_{H_i}(1)\big) \cdot Q_{H_i}(1 - p_i) - \big(F_{G_i}(p_i)- F_{H_i}(p_i)\big) \cdot Q_{H_i}(0)\\
    &-~ \int_{0}^{1 - p_i} (F_{G_i}(p_i + \theta)- F_{H_i}(p_i + \theta)) \cdot Q'_{H_i}(\theta) d\theta \\
    ~\leq&~ (0 - 0) - \int_{0}^{1 - p_i} 0 \cdot Q'_{H_i}(\theta) d\theta ~=~0,
\end{align*}
 where the last line uses $F_{G_i}(v) - F_{H_i}(v) \geq 0$ together with $Q'_{H_i}(\theta) \geq 0$ to bound the integral. 

\paragraph{Bounding \eqref{eq:part2}.} It remains to give a non-negative upper-bound for \eqref{eq:part2}. We start from giving an upper-bound for $Q_{G_i}(\theta) - Q_{H_i}(\theta)$ by discussing two cases based on the value of $S_{G_i}(\theta)$.

\noindent \textbf{Case 1: $S_{G_i}(\theta)$ is large.} To be specific, we show that when $S_{G_i}(\theta) \geq 2 \log \gamma^{-1}$, the value of $Q_{G_i}(\theta)$ is  tiny. To achieve this, we introduce the following claim:

\begin{restatable}{Claim}{clmprodsmall}
\label{clm:prod-small}
    Assume numbers $a_1, \cdots, a_m$ satisfy $a_j \in [0, 1]$ for every $j \in [m]$. Then, we have 
    \[
    \prod_{i \in [m]} (1 - a_j) \leq \exp(-\sum_{i \in [m]} a_j).
    \]
\end{restatable}

We defer the proof of \Cref{clm:prod-small} to \Cref{sec:clmprodsmalldiffsmall}.

Now, we apply \Cref{clm:prod-small} with $m = n$, $a_j = 1 - F_{G_j}(p_j + \theta) $ for $j < i$, $a_i = 0$, and $a_j = 1 - F_{G_j}((p_j + \theta)^-)$ for $j > i$. Then, we have
\begin{align*}
    Q_{G_i}(\theta) ~&\leq~ \exp\left(-\sum_{j \in [m]} a_j\right) \\
    ~&\leq~ \exp\left(-S_{G_i}(\theta) + (1 - F_{G_i}(p_i + \theta))\right) \\
    ~&\leq~ \exp\left(- 2\log \gamma^{-1} + 1\right) ~\leq~ \gamma.
\end{align*}
Therefore, $Q_{G_i}(\theta) - Q_{H_i}(\theta)$ can be bounded by $\gamma$ when $S_{G_i}(\theta) \geq 2 \log \gamma^{-1}$.

\noindent \textbf{Case 2: $S_{G_i}(\theta)$ is small.} Next, we show when $S_{G_i}(\theta) < 2\log \gamma^{-1}$, the difference $Q_{G_i}(\theta) - Q_{H_i}(\theta)$ is bounded. The proof is based on the following claim: 

\begin{restatable}{Claim}{clmdiffsmall}
    \label{clm:diff-small}
    Assume numbers $b_1, \cdots, b_m, c_1, \cdots, c_m$ satisfy $b_j \in [0, 1]$ and $c_j \in [0, b_j]$ for every $j \in [m]$. Then, we have
    \[
    \prod_{j \in [m]} b_j - \prod_{j \in [m]} (b_j-c_j) ~\leq~ \sum_{j \in [m]} c_j.
    \]
\end{restatable}

We defer the proof of \Cref{clm:diff-small} to \Cref{sec:clmprodsmalldiffsmall}.

Now we apply \Cref{clm:diff-small} with $m = n$, $b_j = F_{G_j}(p_j + \theta)$ and $c_j = F_{G_j}(p_j + \theta) - F_{H_j}(p_j + \theta)$ for $j < i$, $b_i = 1$ and $c_i = 0$, and $b_j = F_{G_j}((p_j + \theta)^-)$ and $c_j = F_{G_j}((p_j + \theta)^-) - F_{H_j}((p_j + \theta)^-)$ for $j > i$. Then, we have
\begin{align*}
    Q_{G_i}(\theta) - Q_{H_i}(\theta) ~&\leq~ \sum_{j < i} \left(F_{G_j}(p_j + \theta) - F_{H_j}(p_j + \theta)\right) + \sum_{j > i} \left( F_{G_j}((p_j + \theta)^-) - F_{H_j}((p_j + \theta)^-)\right) \\
    ~&\leq~ \sum_{j < i} \left(\sqrt{(1 -F_{G_j}(p_j + \theta) ) \cdot \gamma} + \gamma \right) + \sum_{j > i} \left(\sqrt{(1 -F_{G_j}((p_j + \theta)^-) ) \cdot \gamma} + \gamma \right) \\
    ~&\leq~ \gamma n + \sqrt{n\cdot \gamma} \cdot \sqrt{\sum_{j < i} \left(1 - F_{G_j}(p_j + \theta) \right) +  \sum_{j > i} \left(1 - F_{G_j}((p_j + \theta)^-) \right)} \\
    ~&=~ \gamma n + \sqrt{S_{G_i}(\theta) \cdot \gamma n} ~\leq~ \gamma n + \sqrt{2\log \gamma^{-1} \cdot \gamma n}\ ,
\end{align*}
where we apply the assumption $F_{G_i}(v) - F_{H_i}(v) ~\leq~ \sqrt{(1 - F_{G_i}(v)) \cdot \gamma} + \gamma$ in the second line, and the Jensen's inequality together with the concavity of function $f(x) = \sqrt{x}$ in the third line. Therefore, after merging with Case 1, we upper-bound $Q_{G_i}(\theta) - Q_{H_i}(\theta)$ as
\begin{align*}
    Q_{G_i}(\theta) - Q_{H_i}(\theta) ~\leq~ \gamma + \one\left[S_{G_i}(\theta) < 2\log \gamma^{-1}\right] \cdot \left(\gamma n + \sqrt{2\log \gamma^{-1} \cdot \gamma n}\right),
\end{align*}
and therefore
\begin{align}
    \text{\eqref{eq:part2}} ~&\leq~ \int_0^{1 - p_i} f_{G_i}(p_i + \theta) \cdot \left(\gamma + \one\left[S_{G_i}(\theta) < 2\log \gamma^{-1}\right] \cdot \left(\gamma n + \sqrt{2\log \gamma^{-1} \cdot \gamma n}\right)\right)  d\theta \notag \\
    ~&\leq~ \gamma + \left(\gamma n + \sqrt{2\log \gamma^{-1} \cdot \gamma n}\right) \cdot \int_{0}^{1 - p_i} f_{G_i}(p_i + \theta) \cdot \one[S_{G_i}(\theta) < 2 \log \gamma^{-1}] d\theta \label{eq:part2-bound}.
\end{align}

\paragraph{Summing the bounds of \eqref{eq:part2}.} Note that \eqref{eq:part2-bound} gives a non-negative upper bound for \eqref{eq:part2}, and recall that \eqref{eq:part1} is upper-bounded by $0$. Therefore, to bound $\rev_{\bG}(\bp) - \rev_{\bH}(\bp)$, it remains to sum \eqref{eq:part2-bound} for each $i \in [n]$, which gives
\begin{align*}
    \rev_{\bG}(\bp) - \rev_{\bH}(\bp) \leq \gamma n  
    + \left(\gamma n + \sqrt{2\log \gamma^{-1}  \gamma n}\right) \cdot \sum_{i \in [n]} \int_{0}^{1 - p_i} f_{G_i}(p_i + \theta)  \one[S_{G_i}(\theta) < 2 \log \gamma^{-1}] d\theta.
\end{align*}

We end the proof of \Cref{lma:approx-sm} by giving the following \Cref{clm:sum-integral-bound}:

\begin{Claim}
    \label{clm:sum-integral-bound}
    For parameter $\beta > 0$, we have
    \[
    \sum_{i \in [n]} \int_{0}^{1 - p_i} f_{G_i}(p_i + \theta)  \one[S_{G_i}(\theta) < \beta] d\theta ~\leq~ \beta + 1. 
    \]
\end{Claim}

\begin{proof}
    Define
    \[
    \ttheta ~:=~ \min\left\{ \theta \geq 0: \sum_{i \in [n]} \left(1 - F_{G_i}(p_i + \theta)\right) < \beta\right\}.
    \]
    Then, for any $\theta \in [0, \ttheta)$ and $i \in [n]$, we have
    \[
    S_{G_i}(\theta) ~\geq~ S_{G_n}(\theta) ~\geq~ \beta,
    \]
    where the first inequality follows from the monotonicity property of $S_{G_i}(\theta)$, and the second inequality follows from the definition of $\ttheta$. Similarly, for any $\theta > \ttheta$ and $i \in [n]$, we have
    \[
    S_{G_i}(\theta) ~\leq~ \sum_{i \in [n]} \mathop{\pr}\limits_{X \sim G_i}[X \geq p_i + \theta]  ~\leq~ \sum_{i \in [n]} \mathop{\pr}\limits_{X \sim G_i}[X > p_i + \ttheta] ~=~ S_{G_n}(\ttheta) < \beta.
    \]
    Therefore, define
    \[
    i^* ~:=~ \min \left\{i \in [n]: S_{G_i}(\ttheta) < \beta\right\},
    \]
    where the feasibility of the above definition comes from the monotonicity of $S_{G_i}(\theta)$ together with the observation that $S_{G_n}(\ttheta) < \beta$. Then, the indicator $\one[S_{G_i}(\theta) < \beta]$ equals to one if and only if either $\theta > \ttheta$, or $\theta = \ttheta$ while $i \geq i^*$. This simplifies our desired integral as
    \begin{align*}
        \sum_{i \in [n]} \int_{0}^{1 - p_i} f_{G_i}(p_i + \theta)  \one[S_{G_i}(\theta) < \beta] d\theta ~&\leq~ \sum_{i < i^*} \int_{\ttheta^+}^{1 - p_i} f_{G_i}(p_i + \theta)  d\theta + \sum_{i \geq i^*} \int_{\ttheta}^{1 - p_i} f_{G_i}(p_i + \theta)  d\theta \\
        ~&=~ \sum_{i < i^*} \left(1 - F_{G_i}(p_i + \ttheta) \right) + \sum_{i \geq i^*} \left(1 - F_{G_i}((p_i + \ttheta)^-) \right) \\
        ~&=~ S_{G_{i^*}}(\ttheta) + \mathop{\pr}\limits_{X \sim G_{i^*}}[X = p_{i^*} + \ttheta] \\
        ~&\leq~ \beta + 1,
    \end{align*}
    where we apply the assumption that $S_{G_{i^*}}(\ttheta) < \beta$ in the last inequality.
\end{proof}

Applying \Cref{clm:sum-integral-bound} to the bound for $\rev_{\bG}(\bp) - \rev_{\bH}(\bp)$ with $\beta = 2\log \gamma^{-1}$, we have
\begin{align*}
    \rev_{\bG}(\bp) - \rev_{\bH}(\bp) ~&\leq~ \gamma n  
    + \left(\gamma n + \sqrt{2\log \gamma^{-1}  \cdot \gamma n}\right) \cdot (2\log \gamma^{-1} + 1) \\
    ~&\leq~ 7\log \gamma^{-1} \cdot \left(\gamma n + \sqrt{\log \gamma^{-1} \cdot \gamma n}\right), 
\end{align*}
where the last inequality follows from rearranging the coefficients of each term.

\section{$n^2/\epsilon^3$ Upper Bound for Query Complexity}
\label{sec:query-upper}

In this section, we prove the $\OTild(n^2 \cdot \epsilon^{-3})$ query complexity upper bound.

\begin{Theorem}
\label{thm:query-upper}
    For \SUDPP problem with $\bD$ being the product value distribution, there exists an algorithm that uses $N = O\left(n^2 \log^9(n/(\epsilon\delta)) \cdot \epsilon^{-3}\right)$ queries from $\bD$ and outputs $\bp = (p_1, \cdots, p_n)$, such that $\rev_{\bD}(\bp) \geq \rev^*_{\bD}  - \epsilon$ holds with probability $1 - \delta$. 
    Furthermore, the algorithm runs in $O(n^{\poly(\epsilon^{-1})})$ time.
\end{Theorem}

\subsection{Warm Up: Ideas for An $\OTild(\poly(n) \cdot \epsilon^{-4})$ Query Complexity}

As a warm-up, we first present some ideas of getting a weaker $\OTild(\poly(n) \cdot \epsilon^{-4})$ query complexity, instead of our $\OTild(\poly(n) \cdot \epsilon^{-3})$ final result. The main idea is to apply Nisan's discretization (e.g., see \cite{BBHM-FOCS05, CHK-EC07}) to reduce the problem to a discretized product distribution, such that the values in the discretized distribution are multiples of $\epsilon^2$, and further run sufficiently many queries for each value in the support set. To be specific, consider rounding the value of product distribution $\bD$ down to the closest multiple of $\epsilon^2$, i.e. for $\bD$, we define the discretized product distribution $\tbD$, such that for $t \in [0, \epsilon^{-2}]$ and $i \in [n]$, we have
\begin{align}
\label{eq:def-td}
    \mathop{\pr}\limits_{X \sim \tD_i}[X = t\cdot \epsilon^2] ~=~ \mathop{\pr}\limits_{X \sim D_i}[t\cdot \epsilon^2 \leq X < (t+1) \cdot \epsilon^2].
\end{align}
Then, the following lemma suggests that assuming $\tbD$ to be the real valuation distribution for the $\SUDPP$ instance only incurs an $O(\epsilon)$ loss:

\begin{restatable}[attributed to Noam Nisan \cite{BBHM-FOCS05, CHK-EC07}]{Lemma}{lmaNisan}
    \label{lma:Nisan}
    Let $\bD$ be a product value distribution and  $\tbD$ be the corresponding discrete distribution defined via \eqref{eq:def-td}. Given any price vector $\bp \in [0, 1]^n$, there exists a poly-time algorithm that outputs another price vector $\tilde \bp$, such that
    \[
    \rev_{\tbD} (\tilde \bp) ~\geq~ \rev_{\bD}(\bp) - \epsilon.
    \]
    Similarly, given any price vector $\tilde \bp \in [0, 1]^n$, there exists a poly-time algorithm that outputs another price vector $\bar \bp$, such that
    \[
    \rev_{\bD} (\bar \bp) ~\geq~ \rev_{\tbD}(\tilde \bp) - \epsilon.
    \]
\end{restatable}

 The proof of the lemma is deferred to Section~\ref{sec:nisan}. 
With \Cref{lma:Nisan}, it suffices to assume that $\tbD$ is the underlying value distribution, as \Cref{lma:Nisan} suggests that replacing the original distribution by a discretized distribution only brings an extra $O(\epsilon)$ loss.  To further give an $\OTild(\poly(n) \cdot \epsilon^{-4})$ query complexity, note that each $\tD_i$ has a support size of $\epsilon^{-2}$. If we query each multiple of $\epsilon^{-2}$ for $\OTild(n \cdot \epsilon^{-2})$ rounds, it can be shown that with high probability the CDF $F_{E_i}(v)$ of corresponding empirical distribution $\bE$ provided by these queries satisfies 
\begin{align}
\label{eq:init-bound}
    |F_{E_i}(v) - F_{\tD_i}(v)| \leq \sqrt{(1 - F_{\tD_i}(v)) \cdot F_{\tD_i}(v)\cdot \frac{\epsilon^2}{n}} + \frac{\epsilon^2}{n},
\end{align}
which is similar to the bound provided by \Cref{lma:Bernstein}. Therefore, by providing a proof similar to the proof of \Cref{thm:sample-upper}, it can be shown that running the PTAS algorithm for empirical distribution $\bE$ is sufficient to give a price vector with revenue $\OTild(\epsilon)$ close to $\rev^*_{\tbD}$. Since the above algorithm requires querying $O(\epsilon^{-2})$ thresholds for all $n$ distributions, and each threshold is queried for $\OTild(n \cdot \epsilon^{-2})$ rounds, the total query complexity is $\OTild(n^2 \cdot \epsilon^{-4})$.

\subsection{Improving the Complexity via Non-Uniform Queries: Proof of \Cref{thm:query-upper}}

Now, we discuss the main ideas for improving the $\OTild(n^2 \cdot \epsilon^{-4})$ query complexity to $\OTild(n^2 \cdot \epsilon^{-3})$ and prove \Cref{thm:query-upper}. 

The first key idea we use is non-uniform queries.  To provide concrete intuition, consider that we are estimating $F_{\tD_i}(v)$, where $v$ is a multiple of $\epsilon^2$ and $F_{\tD_i}(v) \approx 0.5$. To satisfy the constraint in \eqref{eq:init-bound}, we run $\OTild(n \cdot \epsilon^{-2})$ queries and obtain $F_{E_i}(v)$ such that $|F_{E_i}(v) - F_{\tD_i}(v)| \leq O(\epsilon/\sqrt{n})$. Now consider to estimate $F_{\tD_i}(v + \epsilon^2)$. Since there are $\epsilon^{-2}$ thresholds to be queried, the probability mass between $v$ and $v + \epsilon^2$ is, on average, approximately $\epsilon^2$. This is significantly smaller than $\epsilon/\sqrt{n}$ when $\epsilon \ll n^{-0.5}$. That is, when the probability mass between $v$ and $v + \epsilon^2$ is small, it is sufficient to use $F_{E_i}(v)$ as an estimate of $F_{E_i}(v + \epsilon^2)$ instead of running another $\OTild(n \cdot \epsilon^{-2})$ queries with threshold $v + \epsilon^2$. As a generalization of this  idea, after estimating $F_{E_i}(v)$, ideally it's sufficient to use $F_{E_i}(v)$ as approximately $\frac{\epsilon}{\sqrt{n}} \cdot \frac{1}{\epsilon^2} = \frac{1}{\epsilon\cdot \sqrt{n}}$ consecutive thresholds around $v$, which reduce the total number of thresholds to be queried from $\epsilon^{-2}$ to $\sqrt{n} \cdot \epsilon^{-1}$, leading to an $\OTild(n^{2.5} \cdot \epsilon^{-3})$ query complexity.

To further improve the query complexity to $\OTild(n^{2} \cdot \epsilon^{-3})$, we introduce the second idea. Note that the extra $\sqrt{n}$ factor comes from the CDF constraint \eqref{eq:init-bound}. To improve the query complexity, instead of applying this Bernstein-style constraint \eqref{eq:init-bound}, we introduce the following new CDF constraint:
\begin{align}
    \label{eq:new-bound}
    |F_{E_i}(v) - F_{\tD_i}(v)| ~\leq~ \epsilon \cdot (1 - F_{\tD_i}(v)) + \frac{\epsilon}{n}.
\end{align}

Then, the following lemma suggests that the above \eqref{eq:new-bound} is sufficient to guarantee that $\rev^*_{\bE}$ is close to $\rev^*_{\tbD}$:

\begin{restatable}{Lemma}{revclose}
    \label{lma:rev-close-query}
    For any product value distribution $\tbD$ and $\bE$ with support $[0,1]^n$, suppose condition $|F_{E_i}(v) - F_{\tD_i}(v)| ~\leq~ \epsilon \cdot (1 - F_{\tD_i}(v)) + \frac{\epsilon}{n}$ holds for every $i \in [n]$ and $v \in [0, 1]$ for $\epsilon < 0.1$. Then, for any price vector $\bp \in [0,1]^n$, we have
    \[
    \rev_{\bE}(\bp) ~\geq~ \rev_{\tbD}(\bp) - 50\epsilon \cdot \log^2(n/\epsilon).
    \]
\end{restatable}
 We defer the proof of \Cref{lma:rev-close-query} to \Cref{sec:rev-close-query}, as it is similar to the proof of \Cref{lma:rev-close}.

Next, we show that the idea of the non-uniform query is still applicable for the new constraint \eqref{eq:new-bound}. With this new constraint, on average each $F_{E_i}(v)$ can be used as the estimate for approximately $\epsilon \cdot \frac{1}{\epsilon^2} = \epsilon^{-1}$ consecutive thresholds, which further improves the query complexity to the desired $\OTild(n^{2} \cdot \epsilon^{-3})$. To be specific, we give the following \Cref{lma:query-alg}:

\begin{Lemma}
    \label{lma:query-alg}
    For any single-dimensional distribution $G$ with discrete support $\{k\cdot \epsilon^2: k \in \mathbb{Z} \cap [0, \epsilon^{-2}]\}$, there exists an algorithm that runs $O(n \log^3(n/(\epsilon \delta)) \cdot \epsilon^{-3})$ queries and output distribution $H$, such that with probability $1 - \delta/n$, we have
    \[
    |F_{G}(v) - F_{H}(v)| ~\leq~ \epsilon \cdot (1 - F_{G}(v)) + \frac{\epsilon}{n}.
    \]
\end{Lemma}

Applying \Cref{lma:query-alg} to $\tD_1, \cdots, \tD_n$ gives the algorithm that learns $\bE$ with $\OTild(n^{2} \cdot \epsilon^{-3})$ queries. We defer the proof of \Cref{lma:query-alg} to \Cref{sec:query-alg}, and first combine \Cref{lma:query-alg} with \Cref{lma:rev-close-query} to prove \Cref{thm:query-upper}:

\begin{proof}[Proof of \Cref{thm:query-upper}] We assume $\epsilon < 0.1$, otherwise \Cref{thm:query-upper} naturally holds.

We first describe the algorithm for \Cref{thm:query-upper}. Define discretized product distribution $\tbD$ according to \Cref{eq:def-td}. Then, each single-dimensional distribution $\tD_i$ is with support $\{k\cdot \epsilon^2: k \in \mathbb{Z} \cap [0, \epsilon^{-2}]\}$. Next, for every $i \in [n]$, we call the algorithm given by \Cref{lma:query-alg} with $G = \tD_i$, and let $E_i = H$ be the single-dimensional distribution given by the algorithm, and we define $\bE = E_1 \times \cdots \times E_n$ be the corresponding product distribution. 
    
    Next, we run the PTAS algorithm given by \Cref{lma:ptas} for the empirical distribution $\bE$, and get price vector $\bp^{(1)}$ such that $\rev_{\bE}(\bp^{(1)}) \geq \rev^*_{\bE} - \epsilon$. Finally, we apply the algorithm given by \Cref{lma:Nisan}, and outputs a price vector $\bp^{(2)}$ that satisfies $\rev_{\bD}(\bp^{(2)}) \geq \rev_{\tbD}(\bp^{(1)})$. $\bp^{(2)}$ is the final output of our algorithm.

    \noindent \textbf{Proving the success probability.}  Note that when calling the algorithm provided by \Cref{lma:query-alg}, each all succeeds with probability $1 - \delta/n$. By the union bound, with probability $1 - \delta$ we have
    \begin{align}
    \label{eq:query-thm-eq0}
        |F_{E_i}(v) - F_{\tD_i}(v)| ~\leq~ \epsilon \cdot (1 - F_{\tD_i}(v)) + \frac{\epsilon}{n}
    \end{align}
    for all $i \in [n]$ and $v \in [0, 1]$. We will show in the following that \eqref{eq:query-thm-eq0} is the only condition we need. So the algorithm succeeds with probability $1 - \delta$.

    \noindent \textbf{Bounding the query complexity.} Note that we only run queries when calling the algorithm given by \Cref{lma:query-alg}. We call the algorithm $n$ times, and the query complexity of each call is $O(n \log^3(n/(\epsilon \delta)) \cdot \epsilon^{-3})$. Therefore, the total number of queries is $O(n^2 \log^3(n/(\epsilon \delta)) \cdot \epsilon^{-3})$.

    \noindent \textbf{Bounding the error.} We first bound the difference between $\rev_{\tbD}(\bp^{(1)})$ and $\rev^*_{\tbD}$. When \eqref{eq:query-thm-eq0} holds, \Cref{lma:rev-close-query} guarantees
    \begin{align}
        \label{eq:query-thm-eq1}
        \rev^*_{\bE} ~\geq~ \rev_{\bE}(\bp^*_{\tbD}) ~\geq~ \rev_{\tbD}(\bp^*_{\tbD}) - 50\epsilon \cdot \log^2(n/\epsilon) ~=~ \rev^*_{\tbD} - 50\epsilon \cdot \log^2(n/\epsilon).
    \end{align}
    Furthermore, the PTAS algorithm provided by \Cref{lma:ptas} guarantees
    \begin{align}
    \label{eq:query-thm-eq2}
        \rev_{\bE}(\bp^{(1)}) ~\geq~ \rev^*_{\bE} - \epsilon.
    \end{align}
    It remains to bound the difference between $\rev_{\bE}(\bp^{(1)})$ and $\rev_{\tbD}(\bp^{(1)})$. To achieve this, we will apply \Cref{lma:rev-close-query} in a reversed way: Note that \eqref{eq:query-thm-eq0} implies
    \begin{align}
    \label{eq:query-thm-revbound}
        |F_{\tD_i}(v) - F_{E_i}(v)| ~\leq~ 2\epsilon \cdot (1 - F_{E_i}(v)) + \frac{2\epsilon}{n},
    \end{align}
    because if we have $1 - F_{\tD_i}(v) \leq 1/n$, \eqref{eq:query-thm-eq0} implies $|F_{\tD_i}(v) - F_{E_i}(v)| \leq 2\epsilon/n$. Otherwise, when $1 - F_{\tD_i}(v) > 1/n$, \eqref{eq:query-thm-eq0} implies 
    \[
    |F_{\tD_i}(v) - F_{E_i}(v)| ~\leq~ 2\epsilon \cdot (1 - F_{\tD_i}(v)) \leq 0.5 (1 - F_{\tD_i}(v)),
    \]
    where the last inequality holds when $\epsilon < 0.1$, and therefore $1 - F_{E_i}(v) \geq 0.5 (1 - F_{\tD_i}(v))$, implying
    \[
    |F_{\tD_i}(v) - F_{E_i}(v)| ~\leq~ 2\epsilon \cdot (1 - F_{E_i}(v)) + \frac{\epsilon}{n}.
    \]
    Combining both cases proves \eqref{eq:query-thm-revbound}. Then, applying \Cref{lma:rev-close-query} with swapped $\tbD$ and $\bE$, $2\epsilon$ being the error, and $\bp^{(1)}$ being the price vector, we get
    \begin{align}
        \label{eq:query-thm-eq3}
        \rev_{\tbD}(\bp^{(1)}) ~\geq~ \rev_{\bE}(\bp^{(1)}) - 100\epsilon \cdot \log^2(n/(2\epsilon)) ~\geq~ \rev_{\bE}(\bp^{(1)}) - 100\epsilon \cdot \log^2(n/\epsilon).
    \end{align}
    Summing \eqref{eq:query-thm-eq1}, \eqref{eq:query-thm-eq2} and \eqref{eq:query-thm-eq3} together, we get
    \begin{align*}
        \rev_{\tbD}(\bp^{(1)}) ~\geq~ \rev^*_{\tbD} - 150\epsilon \cdot \log^2(n/\epsilon) - \epsilon.
    \end{align*}

    The second step is to bound the difference between $\rev_{\bD}(\bp^{(2)})$ and $\rev^*_{\bD}$. Recall that the second argument in \Cref{lma:Nisan} guarantees
    \[
    \rev_{\bD}(\bp^{(2)}) ~\geq~ \rev_{\tbD}(\bp^{(1)}) - \epsilon.
    \]
    Furthermore, applying the first argument in \Cref{lma:Nisan} gives
    \[
    \rev^*_{\tbD} ~\geq~ \rev_{\tbD}(\tilde \bp) ~\geq~ \rev_{\bD}(\bp^*_{\bD}) - \epsilon ~=~ \rev^*_{\bD} - \epsilon.
    \]
    Combining the above two inequalities with the inequality between $\rev_{\tbD}(\bp^{(1)})$ and  $\rev^*_{\tbD}$ gives
    \[
    \rev_{\bD}(\bp^{(2)}) ~\geq~ \rev^*_{\bD} - 150\epsilon \cdot \log^2(n/\epsilon) - 3\epsilon ~\geq~ \rev^*_{\bD} - 200\epsilon \cdot \log^2(n/\epsilon).
    \]

    \noindent \textbf{Scaling the error parameter.} Note that the above error bound we get is $O(\epsilon \cdot \log^2(n/\epsilon))$, instead of $\epsilon$. To give the final query complexity bound, let $\epsilon' = 200\epsilon \cdot \log^2(n/\epsilon)$ be the scaled error parameter. Then, the total number of queries we need is 
    \begin{align*}
        O \left(\frac{n^2 \log^3(n/(\epsilon\delta))}{\epsilon^3}\right) ~=~ O \left(\frac{n^2 \log^9(n/(\epsilon\delta))}{\epsilon'^3}\right) ~=~ O \left(\frac{n^2 \log^9(n/(\epsilon'\delta))}{\epsilon'^3}\right).
    \end{align*}
    Therefore, we have $O\left(n^2 \log^9(n/(\epsilon\delta)) \cdot \epsilon^{-3}\right)$ to be our final query complexity.
\end{proof}

\subsection{Proof of \Cref{lma:query-alg}}
\label{sec:query-alg}

To prove \Cref{lma:query-alg}, we show \Cref{alg:query-alg} is the desired algorithm for \Cref{lma:query-alg}.

\begin{algorithm}[tbh]
\caption{\textsc{Learning A Single $G$ via Queries}}
\label{alg:query-alg}
\begin{algorithmic}[1]
\State \textbf{input:} Single-dimensional distribution $G$ with support $\{k\cdot \epsilon^2: k \in \mathbb{Z} \cap [0, \epsilon^{-2}]\}$, parameters $\delta, n, \epsilon$.
\State Initialize $k_{1} = \epsilon^{-2}$, step length $\lambda_{1} = \epsilon/n$, $F_H(k_{1}\cdot \epsilon^2) = 1$, and $j = 1$.
\While{$k_{j} > 0$}
\State Set $l \gets 0, r \gets k_{j} - 1$
\While{$l < r$}
\Comment{\textcolor{blue}{Find $k_{j+1}$ such that $F_{G}(k_{j+1} \cdot \epsilon^2)$ is away from $F_{H}(k_{j} \cdot \epsilon^2)$}}
\State Set $m \gets \lfloor(l + r + 1)/2 \rfloor$
\State Estimate $F_{G}(m \cdot \epsilon^2)$ by running $N = C \cdot \frac{n \log (n/(\epsilon \delta))}{\epsilon^2}$ queries with threshold $m \cdot \epsilon^2$. 
\State Let $\hat F_{G}(m \cdot \epsilon^2) = \frac{1}{N} \cdot \sum_{t \in [N]} \one[X_t \leq m \cdot \epsilon^2]$ be the estimate for $F_{G}(m \cdot \epsilon^2)$, where $X_t \sim G$ is the hidden sample from $t$-th query.
\If{$\hat F_{G}(m \cdot \epsilon^2) \leq F_{H}(k_{j} \cdot \epsilon^2) - 0.5 \lambda_{j}$}
\Comment{\textcolor{blue}{$F_{G}(m \cdot \epsilon^2)$ is away from $F_{H}(k_{j} \cdot \epsilon^2)$}}
\State Update $l \gets m$ 
\Else
\State Update $r \gets m - 1$
\EndIf
\EndWhile
\State Set $k_{j+1} \gets l$ 
\Comment{\textcolor{blue}{We have $l = r$ after finishing binary search}}
\State Run $N = C \cdot \frac{n \log (n/(\epsilon \delta))}{\epsilon^2}$ queries with threshold $k_{j+1} \cdot \epsilon^2$. 
\State Let $F_{H}(k_{j+1} \cdot \epsilon^2) = \frac{1}{N} \cdot \sum_{t \in [N]} \one[X_t \leq k_{j+1} \cdot \epsilon^2]$ be the estimate for $F_{G}(m \cdot \epsilon^2)$, where $X_t \sim G$ is the hidden sample from $t$-th query.
\State For $i = k_{j+1} + 1, \cdots, k_{j} - 1$, set $F_{H}(i\cdot \epsilon^2) \gets F_H(k_{j} \cdot \epsilon^2)$.
\State Update $k_{j} \gets k_{j+1}$, $\lambda_{j} \gets \epsilon \cdot (1 - F_{H}(k_{j+1} \cdot \epsilon^2)) + \epsilon/n$, and finally $j \gets j+1$.
\EndWhile
\State Set $R \gets j$ to be the total rounds of binary searches we run.
\State \textbf{output:} Distribution $H$ with $F_H(v)$ as its CDF.
\end{algorithmic}
\end{algorithm}

\noindent \textbf{Overview of \Cref{alg:query-alg}.} For a better understanding, we first give an overview of \Cref{alg:query-alg}.

The main idea of \Cref{alg:query-alg} is to find $R \approx O(\epsilon^{-1})$ number of key thresholds, and show that giving accurate CDF estimates for that $R$ thresholds would be sufficient.

To be more specific, our algorithm starts from $k_1 \cdot \epsilon^2 = 1$  as the first key threshold.  In each round, we find the following $k_{j+1} < k_j$, such that the CDF of the new $k_{j+1} \cdot \epsilon^2$ can't be estimated accurately via $F_H(k_j \cdot \epsilon^2)$. Since the CDF of distribution $G$ is monotone, it's sufficient to find $k_{j+1}$ via a standard binary search. Then, for multiples of $\epsilon^2$ which are between $k_{j+1} \cdot \epsilon^2$ and $k_{j} \cdot \epsilon^2$,  $F_H(k_j \cdot \epsilon^2)$ is sufficiently accurate to estimate their CDFs. The algorithm terminates when reaching $k_R = 0$ as the final key threshold.  The algorithm itself guarantees that every multiple of $\epsilon^2$ has an accurate CDF estimate. We will further show that this algorithm terminates with $R = O(\epsilon^{-1})$. Therefore, the algorithm tests $\OTild(\epsilon^{-1})$ thresholds, while each threshold is tested with $N = \OTild(n \cdot \epsilon^{-2})$ queries, which gives the desired $\OTild(n \cdot \epsilon^{-3})$ query complexity stated in \Cref{lma:query-alg}.

\noindent \textbf{Concentrations for \Cref{alg:query-alg}.} To analyze \Cref{alg:query-alg}, we first provide the following concentration bound:

\begin{restatable}{Claim}{concentration}
    \label{clm:concentration}
    For all $\hat F_{G}(v)$ and $F_{H}(v)$ that are calculated from $N = C \cdot \frac{n \log (n/(\epsilon \delta))}{\epsilon^2}$ queries with $v$ being a threshold, we have 
\begin{align*}
    |\hat F_{G}(v) - F_{G}(v)| \leq 0.1 \left(\epsilon \cdot (1 - F_G(v)) + \frac{\epsilon}{n}\right)~ \text{and} ~~  |F_{H}(v) - F_{G}(v)| \leq 0.1  \left(\epsilon \cdot (1 - F_G(v)) + \frac{\epsilon}{n}\right)
\end{align*}
    holds with probability at least $1 - \delta/n$, when $C$ is sufficiently large.
\end{restatable}

We defer the proof of \Cref{clm:concentration} to \Cref{sec:concentration}. With \Cref{clm:concentration},  we also give the following \Cref{clm:step-size}, which gives the accuracy of the step length $\lambda_j$ in each round of binary search:

\begin{Claim}
    \label{clm:step-size}
    Assuming \Cref{clm:concentration} holds. Then, for every $j \in [R]$, we have 
    \[
    0.9 \cdot \left(\epsilon \cdot (1 - F_{G}(k_j \cdot \epsilon^2)) + \frac{\epsilon}{n}\right) ~\leq~ \lambda_{j} ~\leq~ 1.1 \cdot \left(\epsilon \cdot (1 - F_{G}(k_j \cdot \epsilon^2)) + \frac{\epsilon}{n}\right).
    \]
\end{Claim}

\begin{proof}
    Note that 
    \[
    \lambda_j - \left(\epsilon \cdot (1 - F_{G}(k_j \cdot \epsilon^2)) + \frac{\epsilon}{n}\right) ~=~ \epsilon \cdot \left(F_H(k_j \cdot \epsilon^2) - F_G(k_j \cdot \epsilon^2)\right).
    \]
    Since \Cref{clm:concentration} guarantees $|F_{H}(v) - F_{G}(v)| \leq 0.1  \left(\epsilon \cdot (1 - F_G(v)) + \frac{\epsilon}{n}\right)$, there must be
    \[
    (1 - 0.1\epsilon) \cdot \left(\epsilon \cdot (1 - F_{G}(k_j \cdot \epsilon^2)) + \frac{\epsilon}{n}\right) ~\leq~ \lambda_j ~\leq~ (1 + 0.1\epsilon) \cdot \left(\epsilon \cdot (1 - F_{G}(k_j \cdot \epsilon^2)) + \frac{\epsilon}{n}\right).
    \]
    Then, \Cref{clm:step-size} holds from the fact that $0.1 \epsilon \leq 0.1$.
\end{proof}

\noindent \textbf{Proving the accuracy of $F_H(v)$.} Next, we show that \Cref{alg:query-alg} guarantees
\[
|F_{G}(v) - F_{H}(v)| ~\leq~ \epsilon \cdot (1 - F_{G}(v)) + \frac{\epsilon}{n}.
\]
For every $j \in [R]$, \Cref{clm:concentration} already guarantees
\[
|F_{H}(k_j \cdot \epsilon^2) - F_{G}(k_j \cdot \epsilon^2)| \leq 0.1  \left(\epsilon \cdot (1 - F_G(k_j \cdot \epsilon^2)) + \frac{\epsilon}{n}\right).
\]
Note that \Cref{alg:query-alg} sets $F_H(i \cdot \epsilon^2) = F_H(k_j \cdot \epsilon^2)$ for every $i = k_{j+1} + 1, \cdots, k_j$. Then, for every $i = k_{j+1} , \cdots, k_j - 1$, we already get
\begin{align*}
     F_G(i \cdot \epsilon^2) - F_H(i \cdot \epsilon^2) ~&\leq~  F_G(k_j \cdot \epsilon^2) - F_H(k_j \cdot \epsilon^2)\\
    ~&\leq~ 0.1 \cdot \left(\epsilon \cdot (1 - F_G(k_j \cdot \epsilon^2)) + \frac{\epsilon}{n}\right) \\
    ~&\leq~ 0.1 \cdot\left(\epsilon \cdot (1 - F_G(i \cdot \epsilon^2)) + \frac{\epsilon}{n}\right),
\end{align*}
where we apply the monotonicity of CDF $F_G(v)$ in the first and third inequality. Therefore, it remains to upper-bound $F_H(i \cdot \epsilon^2) - F_G(i \cdot \epsilon^2) = F_H(k_j \cdot \epsilon^2) - F_G(i \cdot \epsilon^2)$. Since CDF $F_G(v)$ is monotone, it's sufficient to only consider $i = k_{j+1} + 1$ when $k_{j+1} + 1 < k_j$. Note that the binary search process guarantees that $i = k_{j+1} + 1$ must be tested when $k_{j+1} + 1 < k_j$, and there must be 
\begin{align}
\label{eq:query-alg-acc1}
    \hat F_G((k_{j+1} + 1) \cdot \epsilon^2) > F_{H}(k_j \cdot \epsilon^2) - 0.5 \lambda_j,
\end{align}
which is guaranteed by Line 9 in \Cref{alg:query-alg}.
Therefore, we have
\begin{align*}
    &~F_H(k_j \cdot \epsilon^2) - F_G((k_{j+1} + 1) \cdot \epsilon^2) \\
    ~\leq&~ F_H(k_j \cdot \epsilon^2) - \hat F_G((k_{j+1} + 1) \cdot \epsilon^2) + 0.1 \left(\epsilon \cdot (1 - F_G((k_{j+1} + 1) \cdot \epsilon^2)) + \frac{\epsilon}{n}\right) \\
    ~\leq~& 0.5 \lambda_j + 0.1 \left(\epsilon \cdot (1 - F_G((k_{j+1} + 1) \cdot \epsilon^2)) + \frac{\epsilon}{n}\right) \\
    ~\leq~& 0.55 \cdot \left(\epsilon \cdot (1 - F_G(k_{j}  \cdot \epsilon^2)) + \frac{\epsilon}{n}\right) + 0.1 \left(\epsilon \cdot (1 - F_G((k_{j+1} + 1) \cdot \epsilon^2)) + \frac{\epsilon}{n}\right) \\
    ~\leq&~ \epsilon \cdot (1 - F_G((k_{j+1} + 1) \cdot \epsilon^2)) + \frac{\epsilon}{n},
\end{align*}
where the second line uses \Cref{clm:concentration}, the third line uses \Cref{eq:query-alg-acc1}, the fourth line uses \Cref{clm:step-size}, and the last line uses the fact that $F_G(v)$ is monotone.

\noindent \textbf{Bounding the number of queries for \Cref{alg:query-alg}.} The final step of proving \Cref{lma:query-alg} is to show that \Cref{alg:query-alg} uses no more than $O(n \log^3(n/(\epsilon \delta)) \cdot \epsilon^{-3})$ queries. Note that the algorithm inside the while loop from Line 3 to Line 17 uses a standard binary search, which tests at most $\log_2 (\epsilon^{-2}) + 2 = O(\log \epsilon^{-1})$ thresholds. Each threshold is tested for $O(n\log (n/(\epsilon \delta)) \cdot \epsilon^{-2}$ rounds. Therefore, it's sufficient to show $R = O(\log (n/(\epsilon\delta)) \cdot \epsilon^{-1})$.

Note that the binary search process in \Cref{alg:query-alg} guarantees if $k_{j+1} \neq 0$, threshold $k_{j+1} \cdot \epsilon^2$ must be tested, and we have
\begin{align}
\label{eq:query-alg-acc2}
    \hat F_G(k_{j+1}  \cdot \epsilon^2) \leq F_{H}(k_j \cdot \epsilon^2) - 0.5 \lambda_j,
\end{align}
which is guaranteed by Line 9 in \Cref{alg:query-alg}. Therefore
\begin{align*}
    1 - F_G(k_{j+1}  \cdot \epsilon^2) ~\geq&~ 1 - \hat F_G(k_{j+1}  \cdot \epsilon^2) - 0.1 \left(\epsilon \cdot (1 - F_G(k_{j+1} \cdot \epsilon^2)) +\frac{\epsilon}{n}\right) \\
    ~\geq&~ 1 - F_{H}(k_j \cdot \epsilon^2) + 0.5 \lambda_j - 0.1 \left(\epsilon \cdot (1 - F_G(k_{j+1} \cdot \epsilon^2)) +\frac{\epsilon}{n}\right) \\
    ~\geq&~ 1 - F_{G}(k_j \cdot \epsilon^2) + (0.45-0.1) \left(\epsilon \cdot (1 - F_G(k_{j} \cdot \epsilon^2)) +\frac{\epsilon}{n}\right) \\
    &-~ 0.1 \left(\epsilon \cdot (1 - F_G(k_{j+1} \cdot \epsilon^2)) +\frac{\epsilon}{n}\right) \\
    ~=&~ (1 +0.35\epsilon) \cdot (1 - F_G(k_j \cdot \epsilon^2)) - 0.1\epsilon \cdot (1 - F_G(k_{j+1} \cdot \epsilon^2)) + \frac{0.25\epsilon}{n},
\end{align*}
where the first inequality uses \Cref{clm:concentration}, the second inequality uses \Cref{eq:query-alg-acc2}, and the third inequality uses \Cref{clm:concentration} and \Cref{clm:step-size}.
Rearrange the above inequality, we have
\begin{align*}
    1 - F_G(k_{j+1}  \cdot \epsilon^2) ~&\geq~ \frac{1 + 0.35\epsilon}{1 + 0.1\epsilon} \cdot (1 - F_G(k_j \cdot \epsilon^2)) + \frac{0.25\epsilon}{(1 + 0.1\epsilon)n} \\
    ~&\geq~ (1 + 0.1\epsilon) \cdot (1 - F_G(k_j \cdot \epsilon^2)) + \frac{\epsilon}{10n},
\end{align*}
where the last inequality follows from  $\epsilon < 1$. 

We first apply the above inequality to $j = 1$ and get $1 - F_G(k_2 \cdot \epsilon^2) \geq \frac{\epsilon}{10n}$. Then, we apply the above inequality to $j = 2, 3, \cdots, R - 2$, which gives
\[
1 - F_{G}(k_{j+1} \cdot \epsilon^2) ~\geq~ (1 + 0.1\epsilon) \cdot (1 - F_G(k_j \cdot \epsilon^2)),
\]
and therefore 
\[
1 ~\geq~ 1 - F_{G}(k_{R-1} \cdot \epsilon^2) \geq (1 + 0.1\epsilon)^{R-3} \cdot (1 - F_G(k_2 \cdot \epsilon^2))~\geq~ (1 + 0.1\epsilon)^{R-3} \cdot \frac{\epsilon}{10n}.
\]
Rearranging the above inequality gives
\[
R ~\leq~ \frac{\log(10n/\epsilon)}{\log(1+0.1\epsilon)} + 3 ~=~ O\left(\frac{\log(n/\epsilon)}{\epsilon}\right).
\]

\bibliographystyle{alpha}
\bibliography{ref.bib}

\appendix

\section{$n/\epsilon^2$ Lower Bound for Sample Complexity}
\label{sec:sample-lower}

In this section, we provide a matching $\Omega(n \cdot \epsilon^{-2})$ lower bound for sample complexity. To be specific, we prove the following theorem:

\begin{Theorem}
\label{thm:sample-lower}
    For \SUDPP problem on product distributions, any learning algorithm requires $\Omega(\frac{n}{\epsilon^2})$ samples to return an expected $\epsilon$-additive approximation.
\end{Theorem}

To prove \Cref{thm:sample-lower}, our main idea is to introduce a hard instance, such that for any learning algorithm with less than $C \cdot \frac{n}{\epsilon^2}$ samples with a sufficiently small $C$, the probability that it incurs an $\Omega(\epsilon)$-additive loss is at least $0.1$. We first introduce a base instance of the problem and further convert it into a hard instance by adding some randomization.

\paragraph{Base Instance.} The base instance we consider is a product distribution $\bG$, which contains $n$ i.i.d. distributions $G$ with $n$ being sufficiently large,  such that
\begin{align}
\label{eq:dm-def}
    G = \begin{cases} 1 &  \text{with probability }~~ \frac{0.5}{n} \\ 0.5 & \text{with probability  }~~ \frac{1}{n} \\ 0 & \text{with probability  }~~ 1 - \frac{1.5}{n}.
\end{cases}~
\end{align}

We first consider the optimal price vector $\bp^*_{\bG}$ for distribution $\bG$. Note that to achieve optimal revenue, each price in $\bp^*_{\bG}$ must be either $0.5$ or $1$. If not, increasing all prices within $[0, 0.5)$ to $0.5$ and $(0.5, 1)$ to $1$ can only increase the revenue. For $q \in [0, 1]$ define $\bp^{(q)}$ to be the price vector such that the first $qn$ item prices of $\bp^{(q)}$ equal to $0.5$, while the remaining prices are set to be $1$, i.e., $q \in [0, 1]$ fraction of items in price vector $\bp^{(q)}$ are with price $0.5$, and $1 - q$ fraction of items are with price $1$. Then, there must be $\bp^*_{\bG} = \bp^{(q^*)}$ for some $q^* \in [0, 1]$.

To find $q^*$, consider the value of $\rev_{\bG}(\bp^{(q)})$. Note that we get revenue $1$ when there exists some $i$ with $v_i = \bp^{(q)}_i = 1$, while the items $i'$ with $\bp^{(q)}_{i'} = 0.5$ must satisfy $v_{i'} \neq 1$ (otherwise the utility becomes $0.5$, which is greater than the utility for $i$). Therefore, we have
\begin{align*}
    \pr\left[\text{revenue is }1 \text{ with } \bp^{(q)}\right] ~&=~ \left(1 - \left(1 - \frac{0.5}{n}\right)^{n-qn}\right) \cdot \left(1 - \frac{0.5}{n}\right)^{qn} \\
    ~&=~ \left(1 - \frac{0.5}{n}\right)^{qn} - \left(1 - \frac{0.5}{n}\right)^{n}  \\
    ~&=~ e^{-0.5q} - e^{-0.5} \pm O(n^{-1}).
\end{align*}
Next, consider the case where we get revenue $0$. This happens when all items $i$ with $\bp^{(q)}_i = 1$ satisfy $v_i \neq 1$, while all items $i'$ with $\bp^{(q)}_{i'} = 0.5$ satisfy $v_{i'} = 0$. Therefore, we have
\begin{align*}
    \pr\left[\text{revenue is }0 \text{ with } \bp^{(q)}\right] ~&=~  \left(1 - \frac{0.5}{n}\right)^{n-qn} \cdot \left(1 - \frac{1.5}{n}\right)^{qn} \\
    ~&=~ e^{-0.5(1-q)} \cdot e^{-1.5q} \pm O(n^{-1}) ~=~ e^{-0.5-q} \pm O(n^{-1}).
\end{align*}
The remaining case is that we get revenue $0.5$. Then,
\begin{align*}
    \pr\left[\text{revenue is }0.5 \text{ with } \bp^{(q)}\right] ~&=~  1 - (e^{-0.5q} - e^{-0.5}) - e^{-0.5-q}  \pm O(n^{-1}).
\end{align*}
For simplicity of the notation, define $t = e^{-0.5q} \in [e^{-0.5}, 1]$. Then, our goal is to find $t$ that maximizes
\[
(t - e^{-0.5}) \cdot 1 + (1 - t + e^{-0.5} - t^2/e^{(0.5)}) \cdot 0.5 \pm O(n^{-1}),
\]
which is equivalent to maximize $-t^2e^{-0.5} + t$. Note that $-t^2e^{-0.5} + t$ is maximized when $t^* = e^{0.5}/2$, which further implies $q^* = \ln 4 - 1 \approx 0.3863$. Note that the above calculation comes with an $O(1/n)$ error, which is the difference between $(1 + c/n)^{n/c}$ and $1/e$ with a constant $c$. Therefore, the real optimal $q$ should be $q^* \pm O(1/n)$. We let $q^*_n = \ln 4 - 1 \pm O(1/n)$ be the real optimal parameter, such that $\rev_{\bG}(\bp^{(q^*_n)})$ maximizes the revenue. When $n$ is sufficiently large, we have $q^*_n \in [0.38, 0.39]$.

\paragraph{Converting the base instance.} Next, we add some randomization to the base instance we described above and show that the new instance is hard to learn. Formally, we give the following definition:

\begin{Definition}[Hard Instance]
\label{def:hard}
Consider a \SUDPP instance with underlying product distribution $\tbG$, such that the first $2q^*_n \cdot n$ items in $\tbG$ are divided into $q^*_n \cdot n$ pairs. Inside each pair, there is exactly one item with value distribution $G$ defined in \eqref{eq:dm-def}, and another item with value distribution $G^L$, such that
\begin{align}
\label{eq:dl-def}
    G^L = \begin{cases} 1 &  \text{with probability }~~ \frac{0.5 - \epsilon}{n} \\ 0.5 & \text{with probability  }~~ \frac{1 + \epsilon}{n} \\ 0 & \text{with probability  }~~ 1 - \frac{1.5}{n},
\end{cases}~
\end{align}
where $\epsilon$ is sufficiently small.
\end{Definition}

We will show that the instance $\tbG$ defined in \Cref{def:hard} is the desired hard instance for \Cref{thm:sample-lower}.

Note that when $\epsilon$ is sufficiently small, the optimal pricing strategy for $\tbG$ is the same as $\tD$, i.e., $q^*_n \cdot  n$ number of items are with price $0.5$, and the remaining are with price $1$. The major difference between $\tbG$ and $\bG$ is that there are $q^*_n \cdot n$ number of items with value distribution $G^L$. Since the probability of generating a $1$ from $G^L$ is slightly smaller than that from $G$, intuitively we should price the item with distribution $G^L$ as $0.5$, and the remaining items with distribution $G$ as $1$. The following lemma formally suggests that if we incorrectly mark the price for a distribution $G^L$ as $1$, and the price for a distribution $G$ as $0.5$, swapping the prices of these two distributions incurs at least $\Omega(\frac{\epsilon}{n})$ gain in the revenue.

\begin{Lemma}
\label{lma:mistake-loss}
    Fix a pair of items, such that one item is with distribution $G^L$, and another one is with distribution $G$. Assume the prices of other items are at least $0.5$. If we incorrectly price the item with distribution $G^L$ at $1$ and the item with distribution $G$ at $0.5$, swapping the prices of these two items gains an extra revenue for at least $0.2 \cdot \frac{\epsilon}{n}$.
\end{Lemma}

\begin{proof}
    We prove \Cref{lma:mistake-loss} via a coupling idea.
    Define distribution $G^{(1)}$ and $G^{(0.5)}$, such that $G^{(1)}$ is the distribution with price $1$, and $G^{(0.5)}$ is the distribution with price $0.5$. Both distributions generate $0$ w.p. $1 - 1.5/n$, $0.5$ w.p. $1/n$, and $1$ w.p. $(0.5 - \epsilon)/n$. For the remaining $\epsilon/n$ probability, we mark it as ``undecided zone''. Then, pricing $(G, G^L)$ at prices $(1, 0.5)$ or $(0.5, 1)$ is the same as allocating $1$ and $0.5$ to the undecided zones of $G^{(1)}$ and $G^{(0.5)}$, and proving \Cref{lma:mistake-loss} is equivalent to show that assigning $1$ to the undecided zone of $G^{(1)}$ while assigning $0.5$ to the undecided zone of $G^{(0.5)}$ yields an extra $0.2 \cdot \frac{\epsilon}{n}$ revenue compared to doing the opposite.

    We first assign $0.5$ to the undecided zone of $G^{(1)}$, and $1$ to the undecided zone of $G^{(0.5)}$, and lower bound the changes of revenue when switching the values of two zones. Let $v^{(1)} \sim G^{(1)}$ and $v^{(0.5)} \sim G^{(0.5)}$. Consider the following three cases:

    \noindent \textbf{Case 1: Both $v^{(1)}$ and $v^{(0.5)}$ are outside of undecided zones.} In this case, the revenue remains unchanged.

    \noindent \textbf{Case 2: $v^{(1)}$ falls in the undecided zone.} In this case, the value of $v^{(1)}$ increases from $0.5$ to $1$. Consider the scenario where the values of remaining items are at most $0.5$. In this case, the revenue increases from at most $0.5$ to $1$. Therefore, the expected gain of the revenue is at least 
    \[
    0.5 \cdot \left(1 - \frac{0.5}{n}\right)^{n-1} ~=~ \frac{1}{2\sqrt{e}} \pm O(n^{-1}).
    \]

    \noindent \textbf{Case 3: $v^{(0.5)}$ falls in the undecided zone, while $v^{(1)}$ does not.} In this case, the value of $v^{(0.5)}$ drops from $1$ to $0.5$. For any realization of values for other items, if $v^{(0.5)}$ initially wins the auction and loses the auction after the value drop, the revenue can only increase, because the price of every other item is at least $0.5$. Otherwise, the revenue remains unchanged, because the winner of the auction does not change. Therefore, in this case, the revenue does not decrease.

    To conclude, the total gain of our revenue is at least
    \[
    \left(\frac{1}{2\sqrt{e}} \pm O(n^{-1}) \right) \cdot \pr\left[\text{Case 2 happens}\right] ~\geq~ 0.2 \cdot \frac{\epsilon}{n},
    \]
    where the last inequality holds when $n$ is sufficiently large.
\end{proof}

\Cref{lma:mistake-loss} suggests that whenever we fail to distinguish between $G^L$ and $G$ for one item pair, we incur an $\Omega(\epsilon/n)$ loss in revenue. Since there are $O(n)$ item pairs to be distinguished, it remains to show that when we don't have sufficiently many samples, there is a constant probability of making a mistake for each pair. Formally, we prove the following:

\begin{Lemma}
\label{lma:hellinger}
    Consider the \SUDPP instance defined in \Cref{def:hard}. Suppose that, for all $q^*_n \cdot n$ pairs of distributions stated in \Cref{def:hard}, the distribution pairs $(G^L, G)$ and $(G, G^L)$ occur with equal probability of $0.5$. Then, if the distribution $\tbG$ is not explicitly given and instead $N = C \cdot \frac{n}{\epsilon^2}$ samples are provided, any algorithm must incorrectly distinguish whether a given distribution pair is $(G^L, G)$ or $(G, G^L)$ for at least $0.05 \cdot n$ pairs in expectation.
\end{Lemma}

\begin{proof}
    We start by bounding the total variation distance between the observations of $N$ samples from distribution pairs  $(G^L, G)$ and $(G, G^L)$ via Hellinger distance. Note that 
    \begin{align*}
        1 - H^2(G, G^L) ~&=~ 1 - \frac{1}{2} \left(\left(\sqrt{\frac{0.5}{n}} - \sqrt{\frac{0.5 - \epsilon}{n}}\right)^2 + \left(\sqrt{\frac{1}{n}} - \sqrt{\frac{1 + \epsilon}{n}}\right)^2\right) \\
        ~&\geq~1 - \frac{1}{2} \left(\left(\frac{2\epsilon}{\sqrt{n}}\right)^2 + \left(\frac{\epsilon}{\sqrt{n}}\right)^2\right) \\
        ~&\geq~ 1 - \frac{3\epsilon^2}{n},
    \end{align*}
    where in the second line we use the fact that $\sqrt{0.5} - \sqrt{0.5 - \epsilon} \leq 2\epsilon$ for $\epsilon \in [0, 0.5]$, and $\sqrt{1 + \epsilon} - 1 \leq \epsilon$ for $\epsilon \in [0, 1]$.
    Therefore, we have
    \[
    1 - H^2\left((G, G^L), (G^L, G)\right) ~=~ \left(1 - H^2(G, G^L)\right) \cdot \left(1 - H^2(G^L, G)\right) ~\geq~ 1 - \frac{6 \epsilon^2}{n},
    \]
    where the first equality follows from \cite{gibbs2002choosing}. Following the same equation, we have the Hellinger distance between $N$ samples from $(G^L, G)$ and from $(G, G^L)$ is at most $\sqrt{\frac{12N\epsilon^2}{n}}$. Using the fact that the total variation distance is upper-bounded by $\sqrt{2}$ times the Hellinger distance (also see \cite{gibbs2002choosing}) together with the assumption that $N = C \cdot \frac{n}{\epsilon^2}$, the total variation distance between $N$ samples from $(G^L, G)$ and $N$ samples from $(G, G^L)$ is at most $\sqrt{12 C}$, which is smaller than $0.25$ when $C$ is sufficiently small.

    Therefore, when $C$ is sufficiently small, for each pair of distributions that needs to be distinguished, any learning algorithm will fail to decide whether the setting is $(G^L, G)$ or $(G, G^L)$ with probability at least $0.25$. Then, any algorithm has at least $0.25$ probability of making a mistake to set the correct price, and the expected number of mistakes is at least $0.25 \cdot q^*_n \cdot n \geq 0.05n$, where the last inequality uses that $q^*_n \geq 0.2$ when $n$ is sufficiently large.
\end{proof}

Finally, combining \Cref{lma:mistake-loss} and \Cref{lma:hellinger} proves \Cref{thm:sample-lower}:

\begin{proof}[Proof of \Cref{thm:sample-lower}]
    We prove \Cref{thm:sample-lower} by showing that for the instance defined in \Cref{def:hard}, any learning algorithm that learns $\tbG$ with at most $C \cdot \frac{n}{\epsilon^2}$ samples incurs at least an expected $\Omega(\epsilon)$ additive loss in the revenue when $C$ is sufficiently small. 
    
    Note that we only need to consider the learning algorithms that price every item with either $0.5$ or $1$. If not, increasing all prices within $[0, 0.5)$ to $0.5$ and $(0.5, 1)$ to $1$ can only increase the revenue. Then, \Cref{lma:hellinger} suggests that any learning algorithm must incorrectly price at least $0.05 n$ pairs in expectation when $C$ is sufficiently small, while \Cref{lma:mistake-loss} suggests that each incorrectly priced pair incurs an additive $\Omega(\frac{\epsilon}{n})$ loss in revenue, where \Cref{lma:mistake-loss} is applicable since we assume the price for each item is at least $0.5$. Therefore, when $C$ is sufficiently small, any learning algorithm must lose at least $\Omega(\epsilon)$ revenue.
\end{proof}

\section{$n^2/\epsilon^3$ Lower Bound for Query Complexity of Ex-ante Pricing}
\label{sec:query-lower}

In this section, we prove the $\Omega(\frac{n^2}{\epsilon^3})$ query complexity lower bound for the ex-ante version of the unit-demand pricing problem. To be specific, we prove the following:

\begin{Theorem}
\label{thm:query-lower}
    For the ex-ante version of the \SUDPP problem on product distributions, any learning algorithm requires $\Omega(\frac{n^2}{\epsilon^3})$ queries to return an expected $\epsilon$-additive approximation.
\end{Theorem}

\paragraph{The hard instance.} To define the hard instance for \Cref{thm:query-lower}, we first introduce a base instance, and further convert it into a hard instance by adding some randomization. The base instance is a product distribution $\bH$, which includes $n$ i.i.d. distributions $H$. The single-dimensional distribution $H$ has support $0, \frac{1}{2}, \frac{1}{2} + \epsilon, \frac{1}{2} + 2\epsilon, \cdots, \frac{3}{4} - \epsilon, \frac{3}{4}$, such that for every $k = 0, 1, \cdots, \frac{1}{4\epsilon}$, we have
\begin{align}
\label{eq:defh}
    \mathop{\pr}\limits_{X \sim H}\left[X \geq \frac{1}{2} + k \cdot \epsilon\right] \cdot \left(\frac{1}{2} + k \cdot \epsilon\right) ~=~ \frac{1}{2n}.
\end{align}
Take $k = 0$ in the above equation, we have $\pr_{X \sim H}[X \geq 0.5] = 1/n$. Since $0.5$ is the lowest non-zero value in the support of $H$, we have  $\pr_{X \sim H}[X = 0] = 1 - 1/n$, which further implies that any non-zero price $p_i$ together with $q_i = \pr_{X \sim H}[X \geq p_i]$ is a feasible solution for the ex-ante program \eqref{prog:exante}, as the sum of $q_i$ can't be more than $1$. Furthermore, note that the above definition of $H$ provides an equal-revenue distribution. Therefore, playing $p_i = \frac{1}{2} + k \cdot \epsilon$ for any $k = 0, 1, \cdots, \frac{1}{4\epsilon}$ together with $q_i = \pr_{X \sim H}[X \geq p_i]$ gives an optimal pricing for \eqref{prog:exante}.

Next, we build the hard instance with product distribution $\tbH$. $\tbH$ contains $n$ random distributions $\tH_1, \cdots, \tH_n$, where each $\tH_i$ is defined via the following random process:
\begin{itemize}
    \item Initiate $\tH_i = H$.
    \item Choose $k_i \in \{0, 1, \cdots, \frac{1}{4\epsilon} - 1\}$ uniformly at random.
    \item Move the point mass of $\tH_i$ at value $\frac{1}{2} + k_i \cdot \epsilon$ to value  $\frac{1}{2} + (k_i + 1) \cdot \epsilon$.
\end{itemize}

Note that distribution $\tH_i$ is still with support $\{0, \frac{1}{2}, \frac{1}{2} + \epsilon, \cdots, \frac{3}{4} - \epsilon, \frac{3}{4}\}$. Since the optimal solution of \eqref{prog:exante} with $\tbH$ being the underlying product distribution must be with a positive price, in this section from this point onward we will only discuss the solution of \eqref{prog:exante} such that $p_i \in \{\frac{1}{2} + k\epsilon: k = 0, 1, \cdots, \frac{1}{4\epsilon}\}$. Furthermore, note that setting $q_i = \pr_{v_i \sim \tH_i}[v_i\geq p_i]$ under the above assumption is a feasible solution of \eqref{prog:exante}, we assume $q_i = \pr_{v_i \sim \tH_i}[v_i\geq p_i]$ and further aim at finding
\[
\text{\eqref{prog:exante}} ~=~ \max_{p_1, \cdots, p_n} \sum_{i =1}^n p_i \cdot \pr_{v_i \sim \tH_i}[v_i\geq p_i]
\]
under the constraint that $p_i \in \{\frac{1}{2} + k\epsilon: k = 0, 1, \cdots, \frac{1}{4\epsilon}\}$. Then, the following lemma suggests that the hard instance $\tbH$ has a unique optimal solution for \eqref{prog:exante}, and any other solution that misses the optimal pricing would incur a significant loss.

\begin{Lemma}
    \label{lma:query-low-loss}
    The optimal solution of \eqref{prog:exante} with product distribution $\tbH$ satisfies $p^*_i = \frac{1}{2} + (k_i + 1) \cdot \epsilon$. Any other solution $\{p'_i\}$ is at least
    $\sum_{i \in [n]} \frac{0.25\epsilon}{n} \cdot \one[p'_i \neq p^*_i]$ away from the optimal solution.
\end{Lemma}

\begin{proof}
    Note that \eqref{eq:defh} implies that for every $k = 0, 1, \cdots, \frac{1}{4\epsilon} - 1$, we have
    \begin{align*}
        \mathop{\pr}\limits_{X \sim H}\left[X \geq \frac{1}{2} + k \cdot \epsilon\right] \cdot \left(\frac{1}{2} + k \cdot \epsilon\right) ~=~ \mathop{\pr}\limits_{X \sim H}\left[X \geq \frac{1}{2} + (k + 1) \cdot \epsilon\right]\cdot \left(\frac{1}{2} + (k+ 1) \cdot \epsilon\right) 
    \end{align*}
    Rearranging the above equality by splitting the probability of $X \geq \frac{1}{2} + k \cdot \epsilon$ into $X \geq \frac{1}{2} + (k+1) \cdot \epsilon$ and $X = \frac{1}{2} + k \cdot \epsilon$, we have
    \begin{align}
    \label{eq:query-lower-bound}
        \mathop{\pr}\limits_{X \sim H}\left[X = \frac{1}{2} + k \cdot \epsilon\right] ~=~ \mathop{\pr}\limits_{X \sim H}\left[X \geq \frac{1}{2} + (k + 1) \cdot \epsilon\right] \cdot \frac{\epsilon}{0.5 + k \epsilon} \in [\frac{0.5\epsilon}{n}, \frac{2\epsilon}{n}]
    \end{align}
    where the upper bound follows from the fact that the non-zero mass of $H$ is $\frac{1}{n}$, and the term $\epsilon/(0.5 + k\epsilon)$ is at most $2\epsilon$. For the lower bound, it follows from the fact that the probability of $X = 0.75$ is $2/(3n)$, which is given by \Cref{eq:defh}, and the term $\epsilon/(0.5 + k\epsilon)$ is at least $\epsilon$.

    Now, we consider the optimal price $p^*_i$ for program \eqref{prog:exante} with $\tH_i$ being the underlying distribution. Comparing to the base distribution $H$, $\tH_i$ shifts the point mass at value $\frac{1}{2} + k_i \cdot \epsilon$ to $\frac{1}{2} + (k_i + 1) \cdot \epsilon$. Note that the revenue for $p \in \{\frac{1}{2}, \frac{1}{2} + \epsilon, \cdots, \frac{3}{4}\} \setminus \{\frac{1}{2} + k_i \cdot \epsilon, \frac{1}{2} +(k_i + 1) \cdot \epsilon\}$ remains unchanged, the revenue for $p = \frac{1}{2} + k_i \cdot \epsilon$ drops, while the revenue for $p = \frac{1}{2} +(k_i + 1) \cdot \epsilon$ increases. Therefore, $p^*_i = \frac{1}{2} +(k_i + 1) \cdot \epsilon$ is the unique optimal price for $\tH_i$.

    Next, we show that if we switch to $p'_i \neq p^*_i$, we lose at least $TBD$ in the objective of \eqref{prog:exante}. Recall that $\pr_{X \sim H}[X = 0.5 + k\epsilon] \geq 0.5 \epsilon/n$. Therefore, for $\tH_i$, the revenue of playing $0.5 + (k+1)\epsilon$ increases by at least
    \[
     \frac{0.5\epsilon}{n} \cdot \left(\frac{1}{2} + (k+1)\epsilon\right) ~\geq~ \frac{0.25\epsilon}{n},
    \]
    i.e., the revenue of playing $0.5 + (k+1)\epsilon$ is at least $0.25\epsilon/n$ larger than all other thresholds. Summing the above inequality for all $i \in [n]$ proves \Cref{lma:query-low-loss}.
\end{proof}

Next, we prove the following \Cref{lma:query-low-distinguish}, which suggests that a learning algorithm must take sufficiently many queries to determine whether the mass of value $\frac{1}{2} + k\epsilon$ is moved to $\frac{1}{2} + (k + 1)\epsilon$ for  distribution $\tH_i$:

\begin{Lemma}
    \label{lma:query-low-distinguish}
    For distribution $\tH_i$, when parameter $k_i$ is unknown, any learning algorithm that queries distribution $\tH_i$ and threshold $\frac{1}{2} + k \cdot \epsilon$ with $k \in \{0, 1, \cdots, \frac{1}{4\epsilon} - 1\}$  for at most $N = C_1 \cdot \frac{n}{\epsilon^2}$ rounds can't decide whether $k_i = k$ with probability at least $\frac{1}{2}$, when $\epsilon \leq 0.1$, $n \geq 2$, and $C_1$ is sufficiently small.
\end{Lemma}

\begin{proof}
    Without loss of generality, we only consider learning algorithms that query $\frac{1}{2} + k \epsilon$ for $k = 0, 1, \cdots, \frac{1}{4\epsilon}$. Note that when $k_i = k$, compared to the base distribution $H$, all the CDF values remain unchanged, except the value of $F_{\tH_i}(0.5 + k \epsilon)$ is at most $2\epsilon/n$ smaller than the value of $F_{H}(0.5 + k \epsilon)$, where the value $2\epsilon/n$ follows from \eqref{eq:query-lower-bound}. Next, we argue that $N= C_1 \cdot \frac{n}{\epsilon^2}$ samples from Bernoulli distribution $Ber(1 - F_{\tH_i}(0.5 + k \epsilon))$ has a sufficiently small total variation distance compared to $N$ samples from Bernoulli distribution $Ber(1 - F_{H}(0.5 + k \epsilon))$ when $C_1$ is sufficiently small, and therefore no learning algorithm can decide whether $k_i = k$.

    For simplicity of the notation, define $\xi_1 = 1 - F_{H}(0.5 + k \epsilon)$, and $\xi_2 = 1 - F_{\tH_i}(0.5 + k \epsilon)$. Recall that we have $\pr_{X \sim H}[X \geq 0.5] = 1/n$ and $\pr_{X \sim H}[X = 0.75] = 2/(3n)$, so we have $\xi_1 \in [\frac{2}{3n}, \frac{1}{n}]$, and $\xi_2 \geq \xi_1 - 2\epsilon/n$. Then, for a single sample, we have
    \begin{align*}
        1 - H^2(Ber(\xi_1), Ber(\xi_2)) ~&=~ 1 - \frac{1}{2} \left(\left(\sqrt{\xi_1} - \sqrt{\xi_2}\right)^2 + \left(\sqrt{1 - \xi_1} - \sqrt{1 - \xi_2}\right)^2\right) \\
    ~&\geq~ 1 - \frac{1}{2} \left(\left(\sqrt{\xi_1} - \sqrt{\xi_1 - 2\epsilon/n}\right)^2 + \left(\sqrt{1 - \xi_1 + 2\epsilon/n} - \sqrt{1 - \xi_1}\right)^2\right) \\
    ~&\geq~ 1 - \frac{1}{2} \left(\left(\sqrt{\frac{2}{3n}} - \sqrt{\frac{2}{3n} - \frac{2\epsilon}{n}}\right)^2 + \left(\sqrt{1 - \frac{1}{n} + \frac{2\epsilon}{n}} - \sqrt{1 - \frac{1}{n}}\right)^2\right) \\
    ~&\geq~ 1 - \frac{1}{2} \left(\frac{2}{n} \cdot 4\epsilon^2 + \frac{16\epsilon^2}{n^2}\right) ~\geq~ 1 - \frac{8\epsilon^2}{n},
    \end{align*}
    where the first inequality uses that $\xi_2 \geq \xi_1 - 2\epsilon/n$, the second inequality uses that $\xi_1 \in [2/(3n), 1/n]$ together with the convexity of function $f(x) = \sqrt{x}$, and the third inequality uses that $\sqrt{1/3} - \sqrt{1/3 - \epsilon} \leq 2\epsilon$ when $\epsilon \leq 0.1$, and $\sqrt{1 - \frac{1}{n} + \frac{2\epsilon}{n}} - \sqrt{1 - \frac{1}{n}} \leq \sqrt{0.5 + 2\epsilon/n} - \sqrt{0.5} \leq 4\epsilon/n$ when $n \geq 2$. Then, the Hellinger distance between $N$ samples from $Ber(\xi_1)$ and from $Ber(\xi_2)$ is at most 
    \[
    \sqrt{1 - \left(1 - \frac{8\epsilon^2}{n}\right)^N} ~\leq~ \sqrt{N \cdot \frac{8\epsilon^2}{n}} ~=~ \sqrt{8 C_1},
    \]
    which is further bounded by $0.5$ when $C_1$ is sufficiently small.

    Therefore, when $C_1$ is sufficiently small and $k_i = k$, any algorithm  can't decide whether the queried threshold $\frac{1}{2} + k \cdot \epsilon$ is from distribution $Ber(\xi_1)$ or $Ber(\xi_2)$ with probability at least $0.5$.
\end{proof}

Finally combining \Cref{lma:query-low-loss} and \Cref{lma:query-low-distinguish} proves \Cref{thm:query-lower}:

\begin{proof}[Proof of \Cref{thm:query-lower}]
    We prove \Cref{thm:query-lower} by showing that any learning algorithm with at most $C_2 \cdot \frac{n^2}{\epsilon^3}$ queries incurs an $\Omega(\epsilon)$ loss when finding the optimal prices for program \eqref{prog:exante} with $\tbH = \tH_1 \times \cdots \times \tH_n$ being the underlying distribution.

    When $C_2$ is sufficiently small, at least $n/2$ distributions $\tH_i$ satisfy the following: at least $\frac{1}{2}$ fraction of the thresholds in $\{\frac{1}{2}, \frac{1}{2} + \epsilon, \cdots, \frac{3}{4} - \epsilon\}$ are queried for at most $C_1 \cdot \frac{n}{\epsilon^2}$ times.  This follows from a standard Pigeonhole Principle. We call those $n/2$ distributions to be ``under-queried'' distributions, and those $\frac{1}{2}$ fraction of the thresholds which are queried for at most $C_1 \cdot \frac{n}{\epsilon^2}$ times to be ``under-queried'' thresholds.
    
    For every under-queried distribution, since \Cref{lma:query-low-distinguish} guarantees that each of its under-queried thresholds is with at least $\frac{1}{2}$ probability that any learning algorithm can't decide whether $k_i$ equals to the corresponding threshold, the expected fraction of not distinguishable thresholds is at least $\frac{1}{4}$. Then, with probability at least $\frac{1}{4}$, the random parameter $k_i$ falls in the fraction of thresholds, such that no algorithm can decide whether $k_i$ equals that threshold, and the optimal strategy can only be a uniformly random guess among the $\frac{1}{4}$ fraction of thresholds, with at most $\frac{1}{2}$ chance to make a correct guess. Therefore, for every under-queried distribution, any algorithm has at least $\frac{1}{8}$ probability of making an incorrect guess of $k_i$. Then, \Cref{lma:query-low-loss} guarantees that the expected loss from incorrectly setting $p'_i \neq p^*_i = 0.5 + (k_i+1)\epsilon$ is at least $\frac{\epsilon}{4n}$. Therefore, the total expected loss is at least
    \[
    \frac{n}{2} \cdot \frac{1}{8} \cdot \frac{\epsilon}{4n} = \Omega(\epsilon). \qedhere
    \]
\end{proof}

\section{Deferred Proofs for Sample Complexity Upper Bound}
\label{sec:appendix-sample-upper}

\subsection{Proof of \Cref{clm:cdf-dtoe}}
\label{sec:dtoe}
\dtoe*

\begin{proof}
    We consider the following four cases, based on the value of $F_{D_i}(v)$:

\noindent \textbf{Case 1: $F_{D_i}(v) \leq 8 \Gamma$.} In this case, if $F_{E_i}(v) \leq F_{D_i}(v)$, there must be $|F_{E_i}(v) - F_{D_i}(v)| \leq 8 \Gamma$. If $F_{E_i}(v) \geq F_{D_i}(v)$, note that $F_{E_i}(v) \leq 8\Gamma + \sqrt{0.25 \cdot 2\Gamma} + \Gamma \leq 0.5$, which is true when $\Gamma \leq 0.01$. Then, there must be $\sqrt{F_{D_i}(v) (1 - F_{D_i}(v)) \cdot 2\Gamma} \leq \sqrt{F_{E_i}(v) (1 - F_{E_i}(v)) \cdot 2\Gamma}$, and we have
\[
|F_{E_i}(v) - F_{D_i}(v)| ~\leq~ \sqrt{F_{D_i}(v) (1 - F_{D_i}(v)) \cdot 2\Gamma} + \Gamma ~\leq~ \sqrt{F_{E_i}(v) (1 - F_{E_i}(v)) \cdot 2\Gamma} + \Gamma.
\]

\noindent \textbf{Case 2: $F_{D_i}(v) \geq 1- 8 \Gamma$.} We omit the proof, as it is identical to Case 1, after swapping $F_{D_i}(v)$ and $1 - F_{D_i}(v)$, $F_{E_i}(v)$ and $1 - F_{E_i}(v)$.

\noindent \textbf{Case 3: $8\Gamma \leq F_{D_i}(v) \leq 0.5$.} In this case, if $F_{E_i}(v) \leq F_{D_i}(v)$, there must be
\[
F_{E_i}(v) ~\geq~ F_{D_i}(v) - \sqrt{F_{D_i}(v) \cdot 2 \cdot 2\Gamma} - \Gamma ~\geq~ \frac{F_{D_i}(v)}{8},
\]
where the first inequality uses $F_{D_i}(v) \leq 0.5$, and the second inequality uses $F_{D_i}(v) \geq 8 \Gamma$.
Therefore, we have
\[
\sqrt{F_{D_i}(v) (1 - F_{D_i}(v)) \cdot 2\Gamma} ~\leq~ \sqrt{8F_{E_i}(v) \cdot (1 - F_{E_i}(v)) \cdot 2\Gamma}.
\]
Then, we have
\[
|F_{D_i}(v) - F_{E_i}(v)| ~\leq~ \sqrt{F_{D_i}(v) (1 - F_{D_i}(v)) \cdot 2\Gamma} + \Gamma ~\leq~ \sqrt{F_{E_i}(v) \cdot (1 - F_{E_i}(v)) \cdot 16\Gamma}  + 8 \Gamma.
\]

On the other hand, when $F_{E_i}(v) \geq F_{D_i}(v)$, either we have $F_{E_i}(v) \leq 0.5$, which implies 
\[\sqrt{F_{D_i}(v) (1 - F_{D_i}(v)) \cdot 2\Gamma} \leq \sqrt{F_{E_i}(v) (1 - F_{E_i}(v)) \cdot 2\Gamma},\] 
or we have $F_{E_i}(v) > 0.5$, implying that $F_{E_i}(v) \leq 0.5 + \sqrt{0.25 \cdot 2\Gamma} + \Gamma \leq 0.75$, where the last inequality holds when $\Gamma \leq 0.01$, and therefore $\sqrt{F_{D_i}(v) (1 - F_{D_i}(v)) \cdot 2\Gamma} \leq \sqrt{F_{E_i}(v) (1 - F_{E_i}(v)) \cdot 4\Gamma}$. Then, we have
\[
|F_{D_i}(v) - F_{E_i}(v)| ~\leq~ \sqrt{F_{D_i}(v) (1 - F_{D_i}(v)) \cdot 2\Gamma} + \Gamma ~\leq~ \sqrt{F_{E_i}(v) \cdot (1 - F_{E_i}(v)) \cdot 4\Gamma}  + \Gamma.
\]

\noindent \textbf{Case 4: $0.5 \leq F_{D_i}(v) \leq 1 - 8\Gamma$.} We omit the proof, as it is identical to Case 3, after swapping $F_{D_i}(v)$ and $1 - F_{D_i}(v)$, $F_{E_i}(v)$ and $1 - F_{E_i}(v)$.

To conclude, inequality $|F_{D_i}(v) - F_{E_i}(v)| ~\leq~ \sqrt{F_{E_i}(v) (1 - F_{E_i}(v)) \cdot 16 \Gamma} ~+~ 8 \Gamma$ holds in all four cases, which proves \Cref{clm:cdf-dtoe}.
\end{proof}

\subsection{Proof of \Cref{clm:cdf-bound-bbd}}
\label{sec:clmcdfboundbbd}
\clmcdfboundbbd*

\begin{proof}
    The first inequality simply holds from the definition of $\bbD$. For the second inequality, combining the condition $|F_{D_i}(v) - F_{E_i}(v)| ~\leq~ \sqrt{F_{D_i}(v) (1 - F_{D_i}(v)) \cdot 2\Gamma} ~+~ \Gamma$ and \eqref{eq:dhat-def} gives
    \begin{align}
    \label{eq:cdf-bound-bbd0}
        F_{\breve D_i}(v) - F_{E_i}(v) ~\leq~ 2\sqrt{F_{D_i}(v)(1 - F_{D_i}(v)) \cdot 2\Gamma} + 2\Gamma ~\leq~2\sqrt{(1 - F_{D_i}(v)) \cdot 2\Gamma} + 2\Gamma
    \end{align}
    To further rewrite the difference as a function of $F_{\breve D_i}(v)$, note that if $1 - F_{D_i}(v) \leq 12\Gamma$, there must be 
    \begin{align}
    \label{eq:cdf-bound-bbd1}
        F_{\breve D_i}(v) - F_{E_i}(v) ~\leq~ 2\sqrt{24} \Gamma + 2\Gamma ~\leq~ 12\Gamma.
    \end{align}
    Otherwise, note that 
    \[
    \sqrt{(1 - F_{D_i}(v)) \cdot 2\Gamma} + \Gamma ~\leq~ \sqrt{\frac{(1 - F_{D_i}(v))^2}{6}} + \frac{1 - F_{D_i}(v)}{12} ~\leq~ \frac{1 - F_{D_i}(v)}{2}.
    \]
    Therefore, we have 
    \begin{align*}
        1 - F_{\breve D_i}(v) ~&\geq~ 1 - F_{D_i}(v) - \left(\sqrt{(1 - F_{D_i}(v)) \cdot 2\Gamma} + \Gamma\right) ~\geq~ \frac{1 - F_{D_i}(v)}{2}
    \end{align*}
    Applying the inequality to \eqref{eq:cdf-bound-bbd0} gives
    \begin{align}
        \label{eq:cdf-bound-bbd2}
        F_{\breve D_i}(v) - F_{E_i}(v) ~\leq~ 2\sqrt{2(1 - F_{D_i}(v)) \cdot 2\Gamma} + 2\Gamma ~=~ 4\sqrt{(1 - F_{D_i}(v)) \cdot \Gamma} + 2\Gamma.
    \end{align}
    Finally, combining \Cref{eq:cdf-bound-bbd1} and \Cref{eq:cdf-bound-bbd2} proves \Cref{clm:cdf-bound-bbd}.
\end{proof}

\subsection{Proof of \Cref{clm:prod-small} and \Cref{clm:diff-small}}
\label{sec:clmprodsmalldiffsmall}

\clmprodsmall*

\begin{proof}
    \Cref{clm:prod-small} follows from the fact that
    \[
    \prod_{j \in [m]} (1 - a_j) ~\leq~ \prod_{j \in [m]} e^{-a_j} ~=~ \exp\left(-\sum_{j \in [m]} a_j\right),
    \]
    where the first inequality follows from the fact that $1 - x \leq e^{-x}$.
\end{proof}

\clmdiffsmall*

\begin{proof}
    We prove this via introduction. The base case $m = 1$ holds from $b_1 - (b_1 - c_1) = c_1$. Assume \Cref{clm:diff-small} holds for $m = k$. For $m = k+1$, we have
    \begin{align*}
        \prod_{j \in [k+1]} b_j - \prod_{j \in [k+1]} (b_j-c_j) ~&=~ c_{k+1} \cdot \prod_{j \in [k]} (b_j-c_j) + b_{k+1} \cdot \left(\prod_{j \in [k]} b_j - \prod_{j \in [k]} (b_j-c_j) \right)  \\
        ~&\leq~ c_{k+1} \cdot 1 1 \cdot \sum_{j \in [k]} c_j  ~=~ \sum_{j \in [k+1]} c_j.
    \end{align*}
    Therefore, \Cref{clm:diff-small} holds for $m = k + 1$, implying that \Cref{clm:diff-small} holds via induction.
\end{proof}

\section{Deferred Proofs for Query Complexity Upper Bound}
\label{sec:appendix-query-upper}

\subsection{Proof of \Cref{lma:Nisan}}
\label{sec:nisan}
\lmaNisan*

\begin{proof}
    Without loss of generality, we assume there is an item in the instance with value $0$ and with price $0$ in both $\bD$ and $\tbD$. Then, we assume the buyer always buys an item for any realization of the values.
    
    \paragraph{First Argument.} For the first argument, we assume without loss of generality that $p_1 \geq p_2 \geq \cdots \geq p_n$. Consider to set $\tilde p_i = (1 - \epsilon) \cdot p_i$. Then, for any realization of values $(v_1, \cdots, v_n) \sim \bD$, let $(\tilde v_1, \cdots, \tilde v_n)$ be the corresponding realization from $\tbD$, i.e., $\tilde v_i$ is the largest multiple of $\epsilon^2$ that satisfies $\tilde v_i \leq v_i$. Assume $i$ is the item that wins the auction for values $(v_1, \cdots, v_n)$ with price vector $\bp$, i.e., $v_i - p_i$ maximizes buyer's utility. Then, when switching to  $(\tilde v_1, \cdots, \tilde v_n)$ with price vector $\tilde \bp$, note that for $j > i$, we have
    \[
    \tilde v_j - \tilde p_j ~\leq~ v_j - (1 - \epsilon) \cdot p_j  ~\leq~ v_i - p_i + \epsilon \cdot p_j ~<~ \tilde v_i - \tilde p_i + \epsilon \cdot (p_j + \epsilon - p_i),
    \]
    where the last inequality uses the fact that $\tilde v_i > v_i - \epsilon^2$, and $\tilde p_i = (1 - \epsilon) \cdot p_i$. 
    
    The above inequality suggests that if the buyer switches the preference from buying $i$ to buying $j$ with a lower price, there must be $p_j \geq p_i - \epsilon$. Then, the collected revenue is at least $(1 - \epsilon) \cdot (p_i - \epsilon) ~\geq~ p_i - \epsilon$, i.e., for any realization of values, switching to $\tbD$ and $\tilde \bp$ only loses $\epsilon$ in the revenue, which implies $\rev_{\tbD} (\tilde \bp) ~\geq~ \rev_{\bD}(\bp) - \epsilon$.

    \paragraph{Second Argument.} Similarly, we assume $\tilde p_1 \geq \tilde p_2 \geq \cdots \geq \tilde p_n$, and consider to set $\bar p_i = (1 - \epsilon) \cdot \tilde p_i$. Let $(v_1, \cdots, v_n) \sim \bD$ be the realization of values from $\bD$, and $(\tilde v_1, \cdots, \tilde v_n)$ be the corresponding realization from $\tbD$. Assume $i$ is the item that wins the auction with values $(\tilde v_1, \cdots, \tilde v_n)$ and price vector $\tilde \bp$. When switching from $\tbD$ to $\bD$ from price vector $\tilde \bp$ to $\bar \bp$, note that for $j > i$, we have
    \[
    v_j - \bar p_j ~<~ \tilde v_j + \epsilon^2 - (1 - \epsilon) \cdot \tilde p_j ~\leq~ \tilde v_i - \tilde p_i + \epsilon^2 + \epsilon \cdot \tilde p_j ~\leq~ v_i - \bar p_i + \epsilon \cdot (\tilde p_j + \epsilon - \tilde p_i).
    \]
    Therefore, if the buyer switches the preference from buying $i$ to buying $j$ with a lower price, there must be $\tilde p_j \geq \tilde p_i - \epsilon$, i.e., the collected revenue is at least $(1 - \epsilon) \cdot (\tilde p_i - \epsilon) \geq p_i - \epsilon$. Then, for any realization of values, switching to $\bD$ and $\bar \bp$ only loses $\epsilon$ in the revenue, which implies $\rev_{\bD} (\bar \bp) ~\geq~ \rev_{\tbD}(\tilde \bp) - \epsilon$.
\end{proof}

\subsection{Proof of \Cref{lma:rev-close-query}}
\label{sec:rev-close-query}

In this section, we prove \Cref{lma:rev-close-query}, which is restated below for convenience:

\revclose*

Our high-level idea is similar to the proof of \Cref{lma:rev-close}. We start from defining the product distribution $\breve \bD$, such that for every $i \in [n]$ and $v \in [0, 1]$, we have
\begin{align}
\label{eq:dhat-def-query}
    F_{\breve D_i}(v) ~:=~ \max\left\{1, F_{\tD_i}(v) + \epsilon \cdot (1 - F_{\tD_i}(v)) + \frac{\epsilon}{n}\right\}
\end{align}
Then, the following lemma guarantees that  $\rev_{\breve \bD}(\bp)$ and $\rev_{\tbD}(\bp)$ are sufficiently close:

\begin{Lemma}
\label{lma:tv-d-dhat-query}
     Let $\tbD$ be a product distribution with support $[0, 1]^n$ and $\bbD$ be the product distribution defined by \eqref{eq:dhat-def-query}. Then, for any price vector $\bp \in [0, 1]^n$, we have $\rev_{\bbD}(\bp) \geq \rev_{\tbD}(\bp) - 2\epsilon$. 
\end{Lemma}

Unlike the proofs for sample complexity, we note that it is impossible to show that the total variation distance between $\tbD$ and $\bbD$ defined in \eqref{eq:dhat-def-query} is bounded by $\OTild(\epsilon)$, as the total variation distance between a single $\tD_i$ and $\breve D_i$ is at least $\Omega(\epsilon)$. Therefore, we provide a more direct proof for \Cref{lma:tv-d-dhat-query}.

\begin{proof}[Proof of \Cref{lma:tv-d-dhat-query}]
    Fix the price vector $\bp$. Without loss of generality, we assume $p_1 \geq p_2 \geq \cdots \geq p_n$.  Define product distribution $\bar \bD$, such that
    \[
    F_{\bar D_i}(v) ~:=~ F_{\tD_i}(v) + \epsilon \cdot (1 - F_{\tD_i}(v)) ~=~ (1 - \epsilon) \cdot F_{\tD_i}(v) + \epsilon.
    \]
    We first prove that $\rev_{\bar \bD}(\bp) \geq \rev(\tbD)(\bp) - \epsilon$ by showing the probability that item $i$ wins the auction with $\bar \bD$ as the value distribution is at least $1 - \epsilon$ times the probability that item $i$ wins the auction with $\tbD$ as the value distribution. 
    Recall that we define $F_D(v^-) = \pr_{X\sim D}[X < v]$ for a single-dimensional distribution $D$. We have
    \begin{align*}
        \mathop{\pr}[i \text{ wins auction with } \bar \bD] ~=&~ \int_{0}^{1-p_i} f_{\bar D_i}(\theta + p_i) \cdot \prod_{j < i} F_{\bar D_j}((p_j + \theta)^-) \cdot \prod_{j > i} F_{\bar D_j}(p_j + \theta) \\
        ~\geq&~ \int_{0}^{1-p_i} (1 - \epsilon)f_{\tD_i}(\theta + p_i) \cdot \prod_{j < i} F_{\tD_j}((p_j + \theta)^-) \cdot \prod_{j > i} F_{tD_j}(p_j + \theta) \\
        ~=&~(1 - \epsilon) \cdot \mathop{\pr}[i \text{ wins auction with } \tbD],
    \end{align*}
    where the second line uses the fact that $f_{\bar D_i}(v) = F'_{\bar D_i}(v) = (1 - \epsilon) \cdot F'_{\tD_i}(v) = f_{\tD_i}(v)$ together with the fact that $F_{\bar D_j}(v) \geq F_{\tD_j}(v)$. Therefore, we have
    \[
    \rev_{\bar \bD}(\bp) = \sum_{i \in [n]} p_i \cdot \mathop{\pr}[i \text{ wins auction with } \bar \bD] \geq\sum_{i \in [n]} p_i \cdot \mathop{\pr}[i \text{ wins auction with } \tbD] \geq \rev_{\tbD}(\bp) - \epsilon.
    \]

    Next, we prove $\rev_{\bbD}(\bp) \geq \rev_{\bar \bD}(\bp) - \epsilon$. Observe that $F_{\breve D_i}(v) = \max\{1, F_{\bar D_i}(v) + \epsilon \cdot n^{-1}\}$. Therefore, the total variation distance between distribution $\breve D_i$ and $\bar D_i$ is bounded by $\epsilon \cdot n^{-1}$, and therefore the total variation distance between product distribution $\bbD$ and $\bar \bD$ is bounded by $\epsilon$. As a corollary, we have $\rev_{\bbD}(\bp) \geq \rev_{\bar \bD}(\bp) - \epsilon$. Combining this inequality with the inequality that $\rev_{\bar \bD}(\bp) \geq \rev_{\tbD}(\bp) - \epsilon$ proves \Cref{lma:tv-d-dhat-query}.
\end{proof}

Next, we upper bound $\rev_{\bbD}(\bp) -\rev_{\bE}(\bp)$. We first give the difference bound between $F_{E_i}(v)$ and $F_{\breve D_i}(v)$ as a function of $F_{\breve D_i}(v)$:

\begin{Claim}
\label{clm:cdf-bound-bbd-query}
    For product distributions $\tbD$ and $\bE$ with support $[0, 1]^n$, assume for every $i \in [n]$ and $v \in [0, 1]$ we have $|F_{\tD_i}(v) - F_{E_i}(v)| ~\leq~ \epsilon \cdot (1 - F_{\tD_i}(v)) + \epsilon/n$. Let $\bbD$ be the product distribution defined by \eqref{eq:dhat-def-query}. Then, we have
    \[
    0 ~\leq~ F_{\breve D_i}(v) - F_{E_i}(v) ~\leq~ 4\epsilon \cdot (1 - F_{\breve D_i}(v)) + \frac{4\epsilon}{n}.
    \]
    for every $i \in [n]$ and $v \in [0, 1]$, assuming $\epsilon < 0.1$.
\end{Claim}

\begin{proof}
    The first inequality simply holds from the definition of $\bbD$. For the second inequality, combining the condition $|F_{\tD_i}(v) - F_{E_i}(v)| ~\leq~ \epsilon \cdot (1 - F_{\tD_i}(v)) + \frac{\epsilon}{n}$ and \eqref{eq:dhat-def-query} gives
    \begin{align}
    \label{eq:cdf-bound-bbd-query-mid}
        F_{\breve D_i}(v) - F_{E_i}(v) ~\leq~ 2\epsilon \cdot (1 - F_{\tD_i}(v)) + \frac{2\epsilon}{n}.
    \end{align}
    Note that if $1 - F_{\tD_i}(v) \leq \frac{1}{n}$, We immediately have $F_{\breve D_i}(v) - F_{E_i}(v) ~\leq~ 4\epsilon/n$. Otherwise, there must be
    \[
    1 - F_{\breve D_i}(v) ~\geq~ (1 - \epsilon) \cdot (1 - F_{\tD_i}(v)) - \frac{\epsilon}{n} ~\geq~ 0.9 \cdot (1 - F_{\tD_i}(v)) - 0.1 \cdot (1 - F_{\tD_i}(v)) ~\geq~ 0.5 \cdot (1 - F_{\tD_i}(v)),
    \]
    where the second inequality uses the assumption that $\epsilon < 0.1$. Applying the above inequality to \Cref{eq:cdf-bound-bbd-query-mid}, we have
    \[
    F_{\breve D_i}(v) - F_{E_i}(v) ~\leq~ 4\epsilon \cdot (1 - F_{\breve D_i}(v)) + \frac{2\epsilon}{n}.
    \]
    Combining two cases, we get $F_{\breve D_i}(v) - F_{E_i}(v) ~\leq~ 4\epsilon \cdot (1 - F_{\breve D_i}(v)) + 4\epsilon/n$.
\end{proof}

With \Cref{clm:cdf-bound-bbd-query}, the following \Cref{lma:approx-sm-query} bounds the difference between $\rev_{\bE}(\bp)$ and $\rev_{\bbD}(\bp)$:

\begin{Lemma}
    \label{lma:approx-sm-query}
    Let $\bG$ and $\bH$ be two product distributions with support $[0, 1]^n$. Assume for every $i \in [n]$ and $v \in [0, 1]$, we have
    \[
    0 ~\leq~ F_{G_i}(v) - F_{H_i}(v) ~\leq~ \gamma \cdot (1 - F_{G_i}(v)) + \gamma/n.
    \]
     Then, for any price vector $\bp$, we have 
    \[
    \rev_{\bH}(\bp) \geq \rev_{\bG}(\bp) - 10 \gamma \cdot \log^2 (n/\gamma).
    \]
\end{Lemma}

We defer the proof of \Cref{lma:approx-sm-query} to \Cref{sec:approx-sm-query}, and first finish the proof of \Cref{lma:rev-close-query}:

\begin{proof}[Proof of \Cref{lma:rev-close-query}]
    Consider to define product distribution $\bbD$ via \eqref{eq:dhat-def-query}. Then, \Cref{lma:rev-close-query} guarantees that $\rev_{\bbD}(\bp) \geq \rev_{\tbD}(\bp) - 2\epsilon$.  Next, since \Cref{clm:cdf-bound-bbd-query} guarantees that $0 \leq F_{\breve D_i}(v) - F_{E_i}(v) \leq 4\epsilon\cdot (1 - F_{\breve D_i}(v)) + 4\epsilon/n$, applying \Cref{lma:approx-sm-query} with $\bG = \bbD$, $\bH = \bE$, and $\gamma = 4\epsilon$ gives
    \[
    \rev_{\bE}(\bp) \geq \rev_{\bbD}(\bp) - 40\epsilon \log^2(n/\epsilon).
    \]
    Combining the above inequality with $\rev_{\bbD}(\bp) \geq \rev_{\tbD}(\bp) - 2\epsilon$, we have
    \[
    \rev_{\bE}(\bp) ~\geq~ \rev_{\tbD}(\bp) - 40\epsilon \log^2(n/\epsilon) - 2\epsilon ~\geq~ \rev_{\tbD}(\bp) - 50\epsilon \log^2(n/\epsilon). \qedhere
    \]
\end{proof}

\subsection{Proof of \Cref{lma:approx-sm-query}}
\label{sec:approx-sm-query}

Without loss of generality, assume $p_1 \leq p_2 \leq \cdots p_n$. To prove \Cref{lma:approx-sm-query}, we follow the approach similar to that used for proving \Cref{lma:approx-sm}, which includes the notations we used for proving \Cref{lma:approx-sm}. Below we restate the notations for convenience: We define 
\[
F_{G_i}(v^-)~:=~ \lim_{x \to v^-} F_{G_i}(x) ~=~ \mathop{\pr}\limits_{X \sim G_i}[X < v]  \quad \text{and} \quad F_{H_i}(v^-)~:=~ \lim_{x \to v^-} F_{H_i}(x) ~=~ \mathop{\pr}\limits_{X \sim H_i}[X < v].
\]
With the notation of $F_{G_i}(v^-)$ and $F_{H_i}(v^-)$, for $\theta \in [0, 1 - p_i]$, we also define
\[
Q_{G_i}(\theta) ~:=~ \prod_{j < i} F_{G_j}(p_j + \theta) \cdot \prod_{j > i} F_{G_j}((p_j + \theta)^-)\quad \text{and} \quad Q_{H_i}(\theta) ~:=~ \prod_{j < i} F_{H_j}(p_j + \theta) \cdot \prod_{j > i} F_{H_j}((p_j + \theta)^-).
\]
to be the probability that item $i$ wins the auction with $\bG$/$\bH$ being the value distribution when the value of item $i$ is $p_i + \theta$, and further define 
\[
P_{G_i} := \int_0^{1 - p_i} f_{G_i}(p_i + \theta) \cdot Q_{G_i}(\theta) d\theta \quad \text{and} \quad P_{H_i} ~:=~ \int_0^{1 - p_i} f_{H_i}(p_i + \theta) \cdot Q_{H_i}(\theta) d\theta
\]
to be the probability that item $i$ wins the auction with $\bG$ or $\bH$ being the value distribution. Finally, we define
\[
S_{G_i}(\theta) ~:=~ \sum_{j \leq i} \left(1 - F_{G_j}(p_j + \theta)\right) + \sum_{j > i} \left(1 - F_{G_j}((p_j + \theta)^-)\right).
\]

To bound the difference between $\rev_{\bG}(\bp)$ and $\rev_{\bH}(\bp)$, note that
\[
\rev_{\bG}(\bp) - \rev_{\bH}(\bp) ~=~ \sum_{i \in [n]} p_i \cdot (P_{G_i} - P_{H_i}) ~\leq~ \sum_{i \in [n]}   \max\{0, P_{G_i} - P_{H_i} \}.
\]
Therefore, it's sufficient to give a non-negative upper bound for each $P_{G_i} - P_{H_i}$. We start from decomposing $P_{G_i} - P_{H_i}$ as
\begin{align}
    P_{G_i} - P_{H_i} ~=&~ \int_0^{1 - p_i} f_{G_i}(p_i + \theta) \cdot Q_{G_i}(\theta) d\theta - \int_0^{1 - p_i} f_{H_i}(p_i + \theta) \cdot Q_{H_i}(\theta) d\theta \notag \\
    ~=&~ \int_0^{1 - p_i} \left(f_{G_i}(p_i + \theta) - f_{H_i}(p_i + \theta)\right) \cdot Q_{H_i}(\theta) d\theta \label{eq:part1-query} \\
    &+ \int_0^{1 - p_i} f_{G_i}(p_i + \theta) \cdot \left( Q_{G_i}(\theta) - Q_{H_i}(\theta)\right)  d\theta \label{eq:part2-query}
\end{align}
Note that we have $\text{\eqref{eq:part1-query}} \leq 0$, where the proof is omitted as it is identical to the proof that shows $\text{\eqref{eq:part1}} \leq 0$. It remains to give an upper bound for \eqref{eq:part2-query}.

To bound \eqref{eq:part2-query}, we first bound $Q_{G_i}(\theta) - Q_{H_i}(\theta)$ via discussing the following two cases. The first case is $S_{G_i}(\theta) \geq 2\log (n/\gamma)$. In this case, we have
\begin{align*}
    Q_{G_i}(\theta) - Q_{H_i}(\theta) ~&\leq~ Q_{G_i}(\theta) \\
    ~&\leq~ \exp \left(-\sum_{j < i} (1 - F_{G_j}(p_j + \theta)) -  \sum_{j > i} (1 - F_{G_j}((p_j + \theta)^-))\right) \\
    ~&=~ \exp\left(- S_{G_i}(\theta) + (1 - F_{G_i}(p_i + \theta))\right) ~\leq~ \exp\left(- S_{G_i}(\theta) + 1\right) ~\leq~ \frac{\gamma}{n},
\end{align*}
where the second inequality follows from \Cref{clm:prod-small}.

The second case is $S_{G_i}(\theta) < 2\log (n/\gamma)$. In this case, we have
\begin{align*}
    Q_{G_i}(\theta) - Q_{H_i}(\theta) ~&\leq~ \sum_{j < i} \left(F_{G_j}(p_j + \theta) - F_{H_j}(p_j + \theta)\right) + \sum_{j > i} \left( F_{G_j}((p_j + \theta)^-) - F_{H_j}((p_j + \theta)^-)\right) \\
    ~&\leq~ \sum_{j < i} \left(\gamma \cdot (1 - F_{G_j}(p_j + \theta)) + \frac{\gamma}{n} \right) + \sum_{j > i} \left(\gamma \cdot (1 - F_{G_j}((p_j + \theta)^-)) + \frac{\gamma}{n} \right) \\
    ~&\leq~ \gamma + \gamma \cdot S_{G_i}(\theta) ~\leq~ \gamma \cdot  3\log (n/\gamma),
\end{align*}
where the first inequality follows from \Cref{clm:diff-small}. Combining with the first case, we have
\[
Q_{G_i}(\theta) - Q_{H_i}(\theta) ~\leq~ \frac{\gamma}{n} + 3\gamma \cdot \log (n/\gamma) \cdot \one[S_{G_i}(\theta) < 2\log (n/\gamma)].
\]
Applying the above inequality to \eqref{eq:part2-query}, we have
\begin{align*}
    \text{\eqref{eq:part2-query}} ~&\leq~ \int_0^{1 - p_i} f_{G_i}(p_i + \theta) \cdot \left(\frac{\gamma}{n} + 3\gamma \cdot \log (n/\gamma) \cdot \one[S_{G_i}(\theta) < 2\log (n/\gamma)]\right) d\theta \\
    ~&\leq~ \frac{\gamma}{n} + 3\gamma \cdot \log (n/\gamma) \cdot \int_0^{1 - p_i} f_{G_i}(p_i + \theta) \cdot \one[S_{G_i}(\theta) < 2\log (n/\gamma)] d\theta.
\end{align*}
Note that the above bound for \eqref{eq:part2-query} is non-negative. Therefore, summing the above bound for every $i \in [n]$ gives an upper bound for $\rev_{\bG}(\bp) - \rev_{\bH}(\bp)$. We have
\begin{align*}
    \rev_{\bG}(\bp) - \rev_{\bH}(\bp) ~&\leq~ n \cdot \frac{\gamma}{n} + 3\gamma \log (n/\gamma) \cdot \sum_{i \in [n]} \int_0^{1 - p_i} f_{G_i}(p_i + \theta) \cdot \one[S_{G_i}(\theta) < 2\log (n/\gamma)] d\theta \\ 
    ~&\leq~ \gamma + 3\gamma \log (n/\gamma) \cdot \left(2\log (n/\gamma) + 1\right) ~\leq~ 10 \gamma \cdot \log^2 (n/\gamma),
\end{align*}
where the second inequality follows from \Cref{clm:sum-integral-bound}.

\subsection{Proof of \Cref{clm:concentration}}
\label{sec:concentration}

\concentration*

\begin{proof}[Proof of \Cref{clm:concentration}]
    We first give a rough bound on the total number of $v$ that we estimated: the support size of $G$ is $\epsilon^{-2}$, so the final index $j$ is at most $\epsilon^{-2}$. In each while loop that starts from Line 3 of \Cref{alg:query-alg}, the binary search algorithm tests at most $\log_2(\epsilon^{-2}) \leq \epsilon^{-1}$ thresholds. Therefore, it's sufficient to show the inequality holds for a single $v$ with probability $1 - \epsilon^3\delta/n$, and applying the union bound over at most $\epsilon^{-3}$ tested $v$ proves \Cref{clm:concentration}.

    Now, we prove $|\hat F_{G}(v) - F_{G}(v)| \leq 0.1 \left(\epsilon \cdot (1 - F_G(v)) + \frac{\epsilon}{n}\right)$ holds with probability $1 - \epsilon^3\delta/n$ for some tested $v$. The proof for $F_H(v)$ is omitted, as the process of getting $F_H(v)$ for some $v = k_j$ is identical to that gets $\hat F_{G}(v)$. Note that $\hat F_{G}(v)$ is an unbiased estimator of $F_{G}(v)$ via $N$ queries. For $t \in [N]$, define mean-zero random variable $Y_t = \one[X_t \leq v] -  F_{G}(v)$, where $X_t \sim G$ is the hidden sample from $t$-th query. Then, we have $E[Y_t^2] = F_{G}(v) \cdot (1 - F_{G}(v)) ~\leq~ 1 - F_{G}(v)$. 
    
    Next, we discuss two cases. The first case is $1 - F_{G}(v) \leq 1/n$. In this case, the concentration  bound $0.1 \left(\epsilon \cdot (1 - F_G(v)) + \frac{\epsilon}{n}\right)$ is upper bounded by $0.2\epsilon/n$. Applying \Cref{thm:Bernstein-bounded} with $\varepsilon = 0.2\frac{\epsilon}{n}\cdot N$, $M = 1$, and $\sigma^2 = \sum_{t \in [N]} \E[Y_t^2] \leq N \cdot (1 - F_{G}(v)) \leq N/n$ gives
    \begin{align*}
        \pr\left[|\hat F_{G}(v) - F_{G}(v)| \geq 0.1 \left(\epsilon \cdot (1 - F_G(v)) + \frac{\epsilon}{n}\right) \right] ~&\leq~ \pr\left[\left|\sum_{t \in [N]} Y_t\right| \geq \varepsilon\right] \\
        ~&\leq~ 2\exp\left(- \frac{\varepsilon^2/2}{\sigma^2 + M\varepsilon/3}\right) \\
        ~&\leq~ 2\exp \left(- \frac{0.02 \cdot \epsilon^2 N^2/n^2}{N/n + 0.2\epsilon N/(3n)}\right) \\
        ~&\leq~ 2\exp \left(- \frac{0.01\cdot  \epsilon^2 N}{n}\right) \\
        ~&=~ 2\exp \left(-0.01 C \cdot \log (n/(\delta \epsilon)) \right) \\
        ~&\leq~ \frac{\delta^3 \epsilon^3}{n^3},
    \end{align*}
    where the last inequality holds when $C \geq 1000$.

    The second case is $1 - F_G(v) > 1/n$. In this case, 
     $0.1 \left(\epsilon \cdot (1 - F_G(v)) + \frac{\epsilon}{n}\right)$ is upper bounded by $0.2\epsilon\cdot (1 - F_G(v))$. Applying \Cref{thm:Bernstein-bounded} with $\varepsilon = 0.2\epsilon\cdot N$, $M = 1$, and $\sigma^2 = \sum_{t \in [N]} \E[Y_t^2] \leq N \cdot (1 - F_{G}(v))$ gives 
    \begin{align*}
        \pr\left[|\hat F_{G}(v) - F_{G}(v)| \geq 0.1 \left(\epsilon \cdot (1 - F_G(v)) + \frac{\epsilon}{n}\right) \right] ~&\leq~ \pr\left[\left|\sum_{t \in [N]} Y_t\right| \geq \varepsilon\right] \\
        ~&\leq~ 2\exp\left(- \frac{\varepsilon^2/2}{\sigma^2 + M\varepsilon/3}\right) \\
        ~&\leq~ 2\exp \left(- \frac{0.02 \cdot \epsilon^2 N^2\cdot (1 - F_G(v))^2}{N\cdot (1 - F_{G}(v)) \cdot (1 + 0.2\epsilon/3)}\right) \\
        ~&\leq~ 2\exp \left(-0.01\epsilon^2 (1 - F_G(v)) \cdot N \right) \\
        ~&=~ 2\exp \left(-0.01 C \log (n/(\delta \epsilon)) \cdot (1 - F_G(v)) n \right) \\
        ~&\leq~ \frac{\delta^3 \epsilon^3}{n^3},
    \end{align*}
    where the last inequality uses $1 - F_G(v) \geq 1/n$, and holds when $C \geq 1000$.

    Combining the two cases, we have 
    \[
    \pr\left[|\hat F_{G}(v) - F_{G}(v)| \leq 0.1 \left(\epsilon \cdot (1 - F_G(v)) + \frac{\epsilon}{n}\right) \right]  ~\geq~ 1 - \frac{\delta^3\epsilon^3}{n^3} ~\geq~ 1 - \frac{\epsilon^3 \delta}{n}. \qedhere
    \]
\end{proof}

\end{document}